\documentclass[11pt]{article}
\usepackage[margin=1in]{geometry}
\usepackage[T1]{fontenc}
\usepackage{lmodern}
\usepackage{microtype}
\usepackage{booktabs} 
\usepackage[ruled]{algorithm2e} 
\usepackage{tikz}
\usepackage{bm}
\usepackage{amsmath,amssymb,amsthm}
\usepackage{thmtools}
\usepackage{thm-restate}
\usepackage[authoryear,round]{natbib}
\usepackage[colorlinks=true,citecolor=blue,linkcolor=blue,urlcolor=blue]{hyperref}
\usetikzlibrary{arrows.meta,positioning,calc}

    \SetAlFnt{\small}
\SetAlCapFnt{\small}
\SetAlCapNameFnt{\small}
\SetAlCapHSkip{0pt}
\IncMargin{-\parindent}
\providecommand{\FIGURE}[3]{\centering#1\caption{#2}#3}
\providecommand{\TABLE}[3]{\caption{#1}#2#3}
\providecommand{\Description}[1]{}
\newcommand{\R}{\mathbb{R}}
\newcommand{\N}{\mathbb{N}}
\newcommand{\Q}{\mathbb{Q}}
\newcommand{\inter}{\operatorname{int}}
\newcommand{\rank}{\operatorname{rank}}
\newcommand{\ev}{\operatorname{ev}}

\newcommand{\mR}{\mathcal{R}}

\newcommand{\TC}{\operatorname{TC}}
\newcommand{\QTC}{\vQ^{\Pi_{\TC_X}}}
\newcommand{\QTCP}{\vQ^{\Pi_{\TC_X}}_{|\partial X}}

\newcommand{\PP}{\mathbb{P}}
\newcommand{\EE}{\mathbb{E}}
\newcommand{\va}{\bm{a}}
\newcommand{\vc}{\bm{c}}
\newcommand{\vn}{\bm{n}}
\newcommand{\vp}{\bm{p}}
\newcommand{\vbasin}{\bm{s}}
\newcommand{\vu}{\bm{u}}
\newcommand{\vvec}{\bm{v}}
\newcommand{\vx}{\bm{x}}
\newcommand{\vy}{\bm{y}}
\newcommand{\vtheta}{\bm{\theta}}
\newcommand{\vomega}{\bm{\omega}}
\newcommand{\vzeta}{\bm{\zeta}}
\newcommand{\vF}{\bm{F}}
\newcommand{\vG}{\bm{G}}
\newcommand{\vQ}{\bm{Q}}
\newcommand{\BFnabla}{\bm{\nabla}}

\theoremstyle{plain}
\newtheorem{theorem}{Theorem}[section]

\newtheorem{proposition}[theorem]{Proposition}
\newtheorem{lemma}[theorem]{Lemma}
\newtheorem{corollary}[theorem]{Corollary}
\newtheorem{example}[theorem]{Example}
\newtheorem{definition}[theorem]{Definition}
\newtheorem{assumption}{Assumption}[section]
\newtheorem{remark}{Remark}[section]
\newenvironment{proofsketch}{\begin{proof}[Proof sketch]}{\end{proof}}

\AtBeginDocument{%
}

\title{Inference From Random Restarts}
\author{%
  Moeen Nehzati\thanks{Department of Economics, New York University. Email: \texttt{moeen.nehzati@nyu.edu}.}
  \and
  Diego Cussen\thanks{Department of Economics, New York University. Email: \texttt{dc5004@nyu.edu}.}
}
\date{\today}

\begin{document}

\maketitle

\begin{abstract}
  Random-restart heuristics are widely used in nonconvex optimization and equilibrium computation: practitioners run a local algorithm from many initial conditions and interpret repeated convergence to the same output as evidence that the result is robust, dominant, or even unique. Despite its widespread use, this reasoning is usually informal. We provide a probabilistic framework for interpreting restart evidence. We give broad, easy-to-verify sufficient conditions under which repeated runs of a solver can be treated as independent draws from a categorical distribution induced by random initial conditions. Within this framework, we develop Bayesian inference from repeated identical outputs. We derive posterior concentration rates for basin size and uniqueness. These rates demonstrate that uniqueness is inherently harder than learning basin size: posterior concentration for uniqueness is polynomial, whereas basin size concentrates exponentially fast. We also provide a verification protocol for checking whether a given problem fits our framework. We demonstrate the protocol on a widely used equilibrium solver for mixed-logit demand with multi-product firms, and complement the verification exercise with posterior tables that apply to any restart experiment satisfying the protocol. We conclude by delineating limits of restart-based inference, including failures induced by solver--problem mismatch and limited visibility of alternative outcomes.

\end{abstract}

\section{Introduction}
In many optimization problems, we do not have access to a solver that is guaranteed to find a global optimum. Instead, we often rely on local methods whose outcomes can depend on the choice of initial conditions. A common response is to run the solver from many initializations. When all runs converge to the same terminal outcome, this is often interpreted as evidence for something more. The following are three common interpretations:

\begin{itemize}
    \item The solution is not driven by the particular initial conditions chosen: it is a robust feature of the problem.
    \item The solution is reproducible and dominant: any run is likely to find the same solution.
    \item The solution is unique or globally optimal: there are no other solutions to the problem.
\end{itemize}

These three interpretations are related but distinct. The first states that the solution is not an artifact of a particular initialization, the second is a stronger claim about the commonality of that solution, and the third takes the additional step to rule out any other solutions.

Each of these interpretations has both a solver-level and a problem-level counterpart. Restart evidence speaks directly to the solver under the restart distribution; upgrading it to a claim about the underlying problem requires additional assumptions linking solver terminal outcomes to genuine solutions.

Though commonplace, these are typically heuristics and lack a formal inferential interpretation: it is not obvious when they are valid, or how many runs are required before the conclusions are credible. This paper provides a probabilistic framework for interpreting repeated convergence under different initial conditions, a pattern that appears across optimization, dynamical systems, and equilibrium computation.

In nonconvex optimization, random restarts are a standard response to limitations of local methods. The multistart heuristic samples trial points, launches a local routine from each point, and reports the best solution found; related metaheuristics such as iterated local search also build on repeated calls to local search routines \citep[Ch.~1, Sec.~1.2]{bertsekas2016nonlinear, torn1989global, lourenco2019iterated}. In these settings, repeated convergence to the same terminal outcome is often read as evidence that the reported solution is robust under the restart distribution.

Call the set of initial conditions that lead to a solution its basin of attraction, with size measured under the law governing initial conditions. Repeatedly observing the same outcome is direct evidence that this outcome has a large basin under the restart distribution and the solver. This connects restart evidence to stability interpretations in dynamical systems: \citet{menck2013basin} advocates basin size as a measure of stability of solutions to nonlinear systems, and \citet{schafer2015decentral,kim2016building,menck2014dead} use repeated convergence to the same solution in power systems as evidence for stability.

In certain applications like equilibrium analysis, a solution having a dominant basin may not be enough, since the existence of another solution with a small basin can weaken the predictive content of the model. Traffic models provide one example: \citet{wang2021optimal, engelson2006congestion} use repeated convergence to the same traffic flow pattern as evidence for uniqueness, a key property for predictive content. Similar concerns arise in applied economics and industrial organization, where formal uniqueness conditions are often unavailable even in well-studied environments such as the differentiated-products demand framework of Berry, Levinsohn, and Pakes \citep{BerryLevinsohnPakes1995}.

Restart-based diagnostics in industrial organization illustrate both the appeal and the danger of this reasoning. In applied demand estimation, \citet{knittel2014estimation} document substantial sensitivity of BLP estimates to starting values and optimization routines even under extensive multi-start exercises, while \citet{dube2012improving} show that careful implementation and many starting values can lead to repeated convergence to the same minimum in benchmark applications. The uniqueness interpretation is explicit in \citet{PakesMcGuire1994}, who report computing equilibria from multiple initial conditions and always converging to the same fixed point, while acknowledging that convergence does not theoretically guarantee uniqueness. Subsequent work shows that such evidence can be misleading: \citet{besanko2010learning} find up to nine distinct equilibria in a model closely related to the Ericson--Pakes \citep{ericson1995markov} framework despite using similar computational methods.

The examples above demonstrate the need to use repeated solver outcomes to speak about global properties of a solver, such as basin size, dominance, and multiplicity of solutions. In contrast to the related approaches reviewed in the next section, this paper restricts attention to a particular restart event: many random initializations have all converged to the same terminal outcome. We ask how this event should change beliefs about two distinct objects: the basin size of the observed terminal outcome and the multiplicity of reachable terminal outcomes. We provide a Bayesian framework for this restart experiment, give conditions under which the resulting evidence is statistically well posed, and characterize posterior concentration as the number of restarts grows. The main message is that basin-size inference and multiplicity inference have different statistical difficulty: posterior belief that the observed basin is large can concentrate exponentially fast, while posterior belief in uniqueness generally concentrates only polynomially because it requires ruling out additional terminal outcomes with small basins.

The rest of the paper is organized as follows. Section~\ref{section:literature} reviews methodological work close to restart-based inference. Section~\ref{sec:strategy} explains why restart evidence is about the solver, not the problem, and motivates our central categorical reduction. Section~\ref{sec:model} gives conditions under which solver behavior is statistically well structured, and Section~\ref{sec:inference} develops the Bayesian framework. Section~\ref{sec:applications} shows how to check whether a concrete application fits the framework, with a mixed-logit equilibrium solver as the worked example, and reports posterior concentration under several priors. Section~\ref{sec:limits} discusses limitations, and Section~\ref{sec:conclusion} concludes.

\section{Related Methods}\label{section:literature}

This section focuses on methodological work close to our setting: learning from repeated runs of a randomly initialized solver, rather than surveying application-specific uses of restarts. The paper sits at the intersection of Bayesian inference for discrete latent structure, geometric analysis of basins of attraction, and the interpretation of algorithmic randomness as data. While each literature contains relevant tools, few directly address the inferential problem induced by repeated random initializations of a general solver. Our contribution is to clarify this gap and provide a framework that is both probabilistically coherent and compatible with algorithmic geometry.

\paragraph{Bayesian Inference For Discrete Latent Distributions}
At a formal level, repeated runs of a randomly initialized solver produce categorical observations drawn from an unknown discrete distribution over outcomes. A natural starting point is therefore Bayesian inference for discrete distributions, including Dirichlet and Dirichlet process priors \citep{ferguson1973bayesian,ghosal2017fundamentals}. These priors are widely used to model uncertainty over unknown probability masses and to allow for an unbounded number of latent categories.

However, directly applying nonparametric Bayesian machinery in our setting is not straightforward. Dirichlet process priors place mass on infinitely many potential atoms, most of which are never observed. In the present context, solver outcomes are not arbitrary labels: they correspond to basins of attraction induced by the geometry of an objective function and the dynamics of an algorithm. As a result, the parameters of a nonparametric prior are not free, but constrained by geometric and dynamical considerations. Our framework makes explicit which aspects of the outcome distribution are identifiable from solver outputs and what assumptions need to be imposed on the solver for obtaining them.

For scalar minimization, a related literature studies optimal stopping rules for random restarts of local solvers \citep{boender1987bayesian, zielinski1981statistical, Boender1995}. It assumes an unknown finite number of local optima, each with positive basin size, and chooses a stopping rule based on the observed local-optimum values. Our work applies to more general problems by establishing when solver outcomes can be interpreted as random variables and when there are finitely many solutions.

\paragraph{Basins of Attraction and Algorithmic Geometry}
A separate literature studies basins of attraction in dynamical systems and optimization. Classical results characterize the stability, regularity, and measure-theoretic properties of basins under gradient flows and related dynamics \citep{hirsch2013differential}. In nonconvex optimization, recent work has analyzed how properties of the objective function and algorithm—such as smoothness, step size, and curvature—determine convergence behavior and basin structure \citep[e.g.,][]{lee2016gradient, ge2015escaping, bhojanapalli2016global}.

This literature provides essential intuition for why basins exist and why their sizes may vary dramatically across solutions. However, it is almost entirely deterministic in nature. Basin geometry is treated as an object to be characterized given full knowledge of the function and dynamics, rather than inferred from finite observations. In contrast, our setting assumes that the researcher observes only the terminal outputs of a solver under random initialization, not the underlying geometry itself. Our contribution is not to refine geometric characterizations of basins, but to use the minimal geometric fact—that basins induce a probability distribution over outcomes under random initialization—as the basis for statistical inference.

\paragraph{The Gap Between Geometry and Inference}
Bridging the two perspectives above is nontrivial. While geometry determines basin sizes, solver outputs alone do not reveal the geometric features that generated them. Conversely, standard Bayesian nonparametric priors require specifying or learning infinitely many parameters corresponding to both observed and unobserved outcomes, without reference to the constraints imposed by algorithmic dynamics. This mismatch creates a gap: many prior specifications do not agree with the geometry of the problem.

Our approach can be viewed as addressing precisely this gap. Rather than attempting to infer full geometric structure or to endow the outcome distribution with unrestricted nonparametric flexibility, we prove when finiteness of solutions holds and use that to build small yet expressive parametric models. We focus on inferential questions that are well posed given the data—such as posterior beliefs about dominance or uniqueness—and characterize how evidence accumulates as the number of solver runs grows. In this sense, our framework can be seen as deliberately intermediate: richer than purely heuristic arguments, but more limited and thus more disciplined than unconstrained nonparametric modeling.




\section{Statistical strategy}
\label{sec:setup}\label{sec:strategy}
\subsection{Problems and Solvers}
Because the data used for inference are generated by a particular solver, the inference necessarily concerns that solver, even when the ultimate goal is to learn about the underlying problem. Extending inference from solver to problem therefore requires understanding the relationship between the two.

While a problem has a set of solutions, a solver produces terminal outcomes over the space of initial conditions, including possible failure outcomes. We call the set of terminal outcomes that the solver can produce its reachable terminal outcomes (RTOs). Ideally, the solver's RTOs coincide with the problem's solution set, but in general they need not. Hence speaking about a problem's solutions based on a solver's observed behavior is not automatic. A solver may produce terminal outcomes that are not genuine solutions, or it may fail to produce some genuine solutions. We call it \emph{sound} if it only produces genuine solutions of the underlying problem, and \emph{complete} if it can produce every genuine solution.

Completeness is integral to the interpretation of restart-based inference as inference about the underlying problem. If a solver is not complete, then some solutions of the problem are systematically excluded from the solver's outcome and invisible to restart-based inference; \citet{besanko2010learning}, for example, show that the Pakes--McGuire algorithm can fail to compute certain classes of equilibria. In contrast, if a solver is complete, then every solution of the problem is represented among the solver's possible outcomes, and hence potentially visible to restart-based inference.

We argue that soundness, in contrast, is not as important for interpretation. Take a complete but unsound solver. Then the solver's RTOs contain all genuine solutions, along with additional spurious outcomes. If a particular genuine solution occurs frequently across runs even in the presence of these spurious outcomes, then after discarding them it will still occur frequently among the remaining solutions.

Additionally, even if a solver is not sound, making it sound is straightforward when we can verify whether a candidate solution is genuine or not. For example, for a solver that produces fixed points of a function, we can plug in the terminal outcome to check whether it is indeed a fixed point. If it is, return it as a solution; if not, return some default failure outcome.

The rest of this subsection discusses completeness in more detail. Completeness is a natural property for many standard iterative methods, in the sense that every genuine solution is a reachable terminal outcome. For example, gradient descent on a smooth function is complete for the problem of finding stationary points, and hence local minima: if initialized at a stationary point or minimum, the gradient vanishes and the algorithm terminates. Similarly, fixed-point iterations are complete for the problem of finding fixed points, since initializing at a fixed point causes the iteration to terminate immediately. Thus, for many commonly used local methods, every solution can in principle be recovered from a suitable initialization.

Nevertheless, completeness is not automatic. It is easiest to see this when the solver's initialization space and outcome space do not coincide with the problem's solution space. For example, consider solving a linear integer programming problem by first solving a linear relaxation and then applying a rounding procedure. Even when this procedure always returns an integer feasible solution (is sound), the rounding step may systematically exclude some feasible or optimal integer solutions. To follow the intuition of the last paragraph, unlike gradient descent or fixed-point iterations, initializing the rounding procedure at a particular integer solution does not necessarily cause the procedure to return that solution.

A different example arises in continuation and homotopy methods. Such methods typically begin from a known solution of an auxiliary problem and then track a solution path as the auxiliary problem is continuously deformed into the target problem. If the target problem admits multiple solutions, some of them may not be connected to any admissible initial solution of the homotopy. Others may lie on branches that bifurcate from or disconnect from the branch being tracked. These solutions are part of the underlying problem, but they can never be produced by the continuation procedure regardless of initialization. In such cases, the solver is not complete with respect to the full solution set of the problem.

Since this paper focuses on inference from solver outputs, our formal results are solver-level. When using them to speak about the underlying problem, completeness is the key additional requirement: methods that can exhaust all points satisfying a necessary condition, like gradient descent or fixed-point iterations, are more likely to be complete, while rounding and path-following methods can systematically miss solutions.

Still, completeness alone is not sufficient to interpret restart-based evidence as evidence about the underlying problem. Even when every solution can be produced by some initialization, certain solutions may remain effectively invisible to random restarts because they require very special initial conditions. For example, when gradient descent is used to search for critical points of a smooth function, saddle points can be reachable in principle but rarely observed under typical random restarts. We return to this connection between solver structure and problem structure later; Section~\ref{sec:limits} formalizes a stronger condition that rules out such failures of visibility.

\subsection{Statistical Reduction}
We begin with a simple analogy. Consider a biologist exploring a newly discovered island. Suppose every mammal observed so far belongs to the same species. The biologist wants to know whether this is the only mammal species on the island, or whether other species exist but have not yet been encountered. While nontrivial, this is a familiar statistical problem: the observations can be modeled as draws from a categorical distribution over the island's finitely many mammal species, and one can infer the prevalence of the observed species or the number of latent categories in various ways.

\begin{figure}[t]
\FIGURE
{\scalebox{0.95}{\begin{tikzpicture}[
  font=\small,
  >=Latex,
  dot/.style={circle,draw=black!25,fill=black,inner sep=0pt,minimum size=3pt},
  dotB1/.style={dot,fill=blue!70},
  dotB2/.style={dot,fill=teal!70},
  dotB3/.style={dot,fill=orange!85!black},
  dotB4/.style={dot,fill=gray!80},
  dotU/.style={dot,fill=black!35},
  box/.style={draw,rounded corners=2pt,thick,fill=black!2},
  part/.style={draw=black!35,line width=0.6pt},
  arr/.style={->,line width=0.45pt,draw=black!40},
]
  \node[box,minimum width=4.6cm,minimum height=4.4cm] (X) {};

  \node[box,minimum width=4.0cm,minimum height=3.3cm, anchor=south west] (Y) at ($(X.south east)+(2.7cm,0.1cm)$) {};

  \node[font=\bfseries] at ($ (X.north)!0.5!(Y.north) + (0,0.80cm) $) {};

  \path[fill=blue!6]   (X.north west) rectangle ($(X.south east)!0.75!(X.north east)$);
  \path[fill=teal!6]   ($(X.south west)!0.75!(X.north west)$) rectangle ($(X.south east)!0.50!(X.north east)$);
  \path[fill=orange!7] ($(X.south west)!0.50!(X.north west)$) rectangle ($(X.south east)!0.25!(X.north east)$);
  \path[fill=gray!7]   ($(X.south west)!0.25!(X.north west)$) rectangle (X.south east);

  \draw[part] ($(X.south west)!0.25!(X.north west)$) -- ($(X.south east)!0.25!(X.north east)$);
  \draw[part] ($(X.south west)!0.50!(X.north west)$) -- ($(X.south east)!0.50!(X.north east)$);
  \draw[part] ($(X.south west)!0.75!(X.north west)$) -- ($(X.south east)!0.75!(X.north east)$);

  \node[
    font=\small\bfseries,
    text=black!85,
    fill=none,
    rounded corners=1pt,
    inner xsep=6pt,
    inner ysep=2pt
  ] at ($(X.north)+(0,-0.45cm)$) {Initial Conditions};
  \node[
    font=\small\bfseries,
    text=black!85,
    fill=white,
    fill opacity=0.85,
    text opacity=1,
    rounded corners=1pt,
    inner xsep=6pt,
    inner ysep=2pt
  ] at ($(Y.north)+(0,-0.45cm)$) {RTO};

  \coordinate (B1c) at ($(X.north)!0.125!(X.south)$);
  \coordinate (B2c) at ($(X.north)!0.375!(X.south)$);
  \coordinate (B3c) at ($(X.north)!0.625!(X.south)$);
  \coordinate (B4c) at ($(X.north)!0.875!(X.south)$);
  \coordinate (xL)  at ($(X.west)!0.04!(X.east)$);

  \node[dotB1] (x11) at ($(B1c)+(-1.45cm, 0.35cm)$) {};
  \node[dotB1] (x12) at ($(B1c)+(-0.10cm,-0.24cm)$) {};
  \node[dotB1] (x14) at ($(B1c)+( 0.90cm,-0.32cm)$) {};

  \node[dotB2] (x21) at ($(B2c)+(-1.30cm, 0.18cm)$) {};
  \node[dotB2] (x22) at ($(B2c)+( 0.20cm,-0.12cm)$) {};

  \node[dotB3] (x31) at ($(B3c)+(-1.40cm, 0.10cm)$) {};
  \node[dotB3] (x32) at ($(B3c)+( 0.05cm,-0.12cm)$) {};
  \node[dotB3] (x33) at ($(B3c)+( 1.25cm, 0.16cm)$) {};


  \node[anchor=west,text=black!90] at (xL |- B1c) {\small $B_0$};
  \node[anchor=west,text=black!90] at (xL |- B2c) {\small $B_1$};
  \node[anchor=west,text=black!90] at (xL |- B3c) {\small $B_2$};
  \node[anchor=west,text=black!90,font=\bfseries] at (xL |- B4c) {\small $\vdots$};

  \node[dot,draw=teal!55,fill=teal!70] (s1) at ($(Y.north)!0.42!(Y.south)$) {};
  \node[dot,draw=orange!65!black,fill=orange!85!black] (s2) at ($(Y.north)!0.62!(Y.south)$) {};
  \node[anchor=west,text=black!90] at ($(s1)+(0.12cm,0)$) {\small $\vy_1$};
  \node[anchor=west,text=black!90] at ($(s2)+(0.12cm,0)$) {\small $\vy_2$};

  \node[dotU] at ($(Y.center)+(-1.10cm, 0.95cm)$) {};
  \node[dotU] at ($(Y.center)+( 1.15cm, 0.75cm)$) {};
  \node[dotU] at ($(Y.center)+(-0.95cm,-1.00cm)$) {};
  \node[dotU] at ($(Y.center)+( 1.05cm,-0.70cm)$) {};
  \node[dotU] at ($(Y.center)+( 0.20cm, 1.15cm)$) {};

  \node[circle,draw=black!55,line width=0.7pt,fill=black!2,inner sep=2.0pt] (na) at ($(Y |- B1c)+(0,0.1cm)$) {\small $\dagger$};

  \draw[arr,draw=blue!55] (x11) -- (na);
  \draw[arr,draw=blue!55] (x12) -- (na);
  \draw[arr,draw=blue!55] (x14) -- (na);

  \draw[arr,draw=teal!55] (x21) -- (s1);
  \draw[arr,draw=teal!55] (x22) -- (s1);

  \draw[arr,draw=orange!65!black] (x31) -- (s2);
  \draw[arr,draw=orange!65!black] (x32) -- (s2);
  \draw[arr,draw=orange!65!black] (x33) -- (s2);

  \draw[rounded corners=2pt,thick] (X.north west) rectangle (X.south east);
  \draw[rounded corners=2pt,thick] (Y.north west) rectangle (Y.south east);
\end{tikzpicture}}}
{A solver $g$ maps an initial condition into a reachable terminal outcome (RTO) or to $\dagger$, which denotes non-convergence or an invalid terminal outcome. Equivalence of initial conditions under $g$ partitions the initial-condition space into basins of attraction, potentially infinitely many of them, denoted by $\vdots$ in the figure.\label{fig:basins-to-rto}}
{}
\end{figure}

Figure~\ref{fig:basins-to-rto} motivates our central reduction. When the terminal-outcome map is measurable, its output becomes a discrete random variable. If it takes only finitely many RTOs, the solver geometry induces a categorical distribution over the terminal outcomes. We call this the \emph{categorical reduction}.

Restart-based inference would be similarly straightforward if solver outcomes automatically had the same structure: draws from a categorical distribution over a finite outcome space. Unfortunately, this is not automatic. First, in our setting, the random object with a known law is the initialization $\vx\sim\mu$. To interpret restart outcomes statistically, the terminal-outcome map must be measurable so that $\vy=g_{\vQ}(\vx)$ is a well-defined random variable and the induced pushforward law on outcomes exists.

Second, even if $g_{\vQ}$ is measurable, the latent outcome structure induced by random restarts can still be complicated: the image $g_{\vQ}(X)$ may be infinite, and the induced basin partition of $X$ can have rich geometry. In such settings, even formulating an ``outcome count'' is delicate, and inference from finitely many runs generally does not reduce to standard categorical models.

In particular, repeated convergence to a single observed outcome certifies only that this outcome has positive restart probability; it does not rule out infinitely many additional outcomes with small (possibly total) probability mass.
This motivates why we separate inference about the basin mass of an observed outcome (which becomes straightforward once $Y$ is a random variable) from inference about multiplicity or outcome counts (which requires additional structural control such as finiteness).

The following two examples illustrate these pathologies.

\begin{example}[Nonmeasurable terminal-outcome map]
Let $X=[0,1]$ equipped with its Borel $\sigma$-algebra and let $\mu$ be a Borel probability measure. Let $S\subseteq[0,1]$ be a non-Borel set and define
\[
g:X\to\{0,1\},\qquad g(x)=\mathbb{I}(x\in S).
\]
If $x\sim\mu$, then the terminal outcome $Y=g(x)$ is not a Borel random variable, since $g^{-1}(\{1\})=S$ is not Borel measurable. Hence there is no probability law on $\{0,1\}$ induced by the original experiment $x\sim\mu$ followed by $Y=g(x)$: such a law $\nu$ would have to satisfy $\nu(A)=\mu(g^{-1}(A))$ for every Borel set $A\subseteq\{0,1\}$, but this expression is not defined for $A=\{1\}$.

One might try to ignore the map $g$ and instead infer a law for $Y$ from repeated evaluations. However, the convergence of empirical frequencies is itself a probabilistic assertion normally guaranteed by measurability.

Since $Y$ is not measurable, the law of large numbers does not apply to $\frac1n\sum_{i=1}^n \mathbb{I}\{g(x_i)=1\}$.
Thus repeated evaluations of the terminal-outcome map need not determine a consistent probability law on $\{0,1\}$. Any distribution placed directly on $Y$ would therefore be an additional modeling assumption, not a distribution justified by the original stochastic system.
\end{example}

\begin{example}[Infinitely many outcomes]
Let $\va\in\R^k$ be nonzero and $b\in\R$. Define
\[
f_{\va,b}:\R^k\to\R,\qquad f_{\va,b}(\vx)=(\va^\top \vx-b)^2.
\]
Then $\arg\min f_{\va,b}=\{\vx\in\R^k:\;\va^\top \vx=b\}$ is an affine hyperplane of dimension $k-1$, hence uncountable. In particular, any solver whose terminal outcome records which minimizer on this set is reached can have continuum-many possible terminal outcomes under random restarts.
\end{example}

Our approach is therefore intentionally front-loaded.

We methodically simplify the solver geometry by giving sufficient conditions under which restart outcomes admit the categorical reduction: a measurable terminal-outcome map with a finite (or effectively finite) set of terminal outcomes. Once these structural conditions are verified, restart outcomes can be treated as categorical data, and inference about basin sizes and solution multiplicity can be carried out using standard Bayesian machinery. Section~\ref{sec:model} develops the structural results; Section~\ref{sec:inference} develops the corresponding inference conditional on them.

\section{Model}\label{sec:model}
Section~\ref{sec:strategy} explained why restart-based inference first needs a categorical reduction from solver geometry to categorical data: the terminal-outcome map must be measurable, and the set of reachable terminal outcomes must be finite, or effectively finite at the relevant observational resolution. This section studies solver-map properties that deliver those requirements.

Practitioners who can verify these assumptions directly may skip to Section~\ref{sec:inference}. The remainder of this section gives our route for dynamic solvers. As discussed in Section~\ref{sec:strategy}, these are solver-level properties; Section~\ref{sec:limits} discusses links to problem-level conclusions.

\subsection{Measurable Dynamic Solvers}
We first formalize dynamic solvers. Intuitively, dynamic solvers start with an initial guess and update it continuously according to a fixed updating rule. Many common solvers, including gradient descent and fixed-point iterations, fall into this category. Their fixed dynamic structure lets us verify the measurability part of the categorical reduction under mild conditions.

We model ideal solvers without numerical imprecision; Subsection~\ref{subsec:numerical} then adds observation maps for numerical error.

Dynamic solvers are modeled as projected autonomous dynamical systems: a system starting from an initial guess evolves continuously according to the projection of a fixed rule $\vQ$. This continuous-time idealization is a tool: the inference in Section~\ref{sec:inference} uses only the terminal map $g_{\vQ}$, its reachable terminal outcomes (RTO), and their basins of attraction, not whether the solver is implemented in discrete or continuous time.

Formally, let $X\subseteq \mathbb{R}^d$ be the space of possible solutions and let $\mu$ be a probability measure on $(X,\mathcal{B}_X)$, where $\mathcal{B}_X$ is the Borel $\sigma$-algebra on $X$. Initial conditions are drawn i.i.d.\ from $\mu$. Let $\lambda_X$ denote $d$-dimensional Lebesgue measure restricted to $X$. We assume $\lambda_X \ll \mu$. Equivalently, for any measurable $A\subseteq X$, if $\lambda_X(A)>0$ then $\mu(A)>0$. In particular, any set of initial conditions with positive Lebesgue measure has positive probability under $\mu$, and therefore will be visited almost surely under infinitely many random restarts.

\begin{assumption}\label{ass:convex-closed}
    $X \subseteq \R^d$ is closed and convex with non-empty interior.
\end{assumption}
As algorithms may not converge to any valid solution (they may oscillate, etc.), we introduce an absorbing graveyard state $\dagger$ to represent non-convergence/non-existence of a valid terminating solution. Thus the space of all possible outputs of a solver is $\tilde X = X\sqcup\{\dagger\}$ (equipped with the disjoint union $\sigma$-algebra). Before finally defining a solver, we need the following definitions:

\begin{definition}[Tangent and Normal Cones]
For closed and convex $X$ (Assumption~\ref{ass:convex-closed}), the tangent cone to $X$ at $\vx$ is defined by:
\[
\TC_X(\vx)
:= \overline{
\Bigl\{
\lambda (\vy-\vx) \mid \vy \in X,\ \lambda \ge 0
\Bigr\}.
}
\]
As this is a closed convex cone, projection onto it is well defined. Let $\Pi_{\TC_X(\vx)}$ denote the projection operator onto $\TC_X(\vx)$.

Dually, we can define the normal cone to $X$ at $\vx$ by:
\[
N_X(\vx) := \{ \vvec \in \R^d \mid \forall \vy \in X:\; \langle \vvec, \vy - \vx \rangle \leq 0 \}
\]

For convex sets, normal cones and tangent cones are polars of each other: $N_X(\vx) = \TC_X(\vx)^\circ$
\end{definition}

One can think of $\TC_X(\vx)$ as the set of directions that point "inside" $X$ at $\vx$ and $N_X(\vx)$ as the normal vectors of all the possible supporting hyperplanes at $\vx$.

Using this, we can now define projected dynamics:

\begin{definition}[Projected Autonomous Dynamics]
    Let $X\subset\mathbb{R}^d$ be closed and convex (Assumption~\ref{ass:convex-closed}), and let
    $\vQ:X\to\mathbb{R}^d$ be a continuous vector field capturing the solver dynamics. For $\vx\in X$, define the
    tangent-cone projected vector field
    \[
        \QTC(\vx) \;:=\;\Pi_{\TC_X(\vx)}\,\vQ(\vx),
    \]
    where $\TC_X(\vx)$ denotes the tangent cone of $X$ at $\vx$.
        
    Given an initial condition $\vx_0\in X$, a projected autonomous dynamical system $\vx$ satisfies
    \[
        \dot \vx(t) \;=\; \QTC(\vx(t)), \qquad \vx(0)=\vx_0.
    \]
    Dually, this can be stated as a differential inclusion concerning the normal cone:
    \[
        \vQ(\vx(t)) \in N_X(\vx(t)) + \dot \vx(t) , \qquad \vx(0)=\vx_0.
    \]
\end{definition}
\begin{remark}
    The projection operator $\Pi_{\TC_X(\vx)}$ is there to ensure the dynamics never leave the space of possible solutions $X$.
\end{remark}

\begin{remark}[$\QTC(\vx)=\vQ(\vx)$ on $\inter X$]
    On the interior of $X$, the projection operator $\Pi_{\TC_X(\vx)}$ is the identity map. Hence we can write:
    \[\QTC(\vx) = \begin{cases}
        \vQ(\vx) & \vx \in \inter X \\
        \Pi_{\TC_X(\vx)} \vQ(\vx) & \vx \in \partial X
    \end{cases}
    \]
    So if $X$ is open (e.g., $\mathbb{R}^d$), then $\QTC = \vQ$.
\end{remark}

Our end goal is to work with maps that are limits of such solutions as $t \to \infty$. Since some dynamics may be locally well defined but fail to exist for all time, we allow maximal local solution maps and record finite-time failure as a non-convergent solver outcome.
\begin{assumption}[Well-Behaved Projected Dynamics]\label{ass:well-behaved-proj}
    For the projected system induced by $\vQ$, each $\vx_0\in X$ has a
    maximal existence time $\tau(\vx_0)\in(0,\infty]$ and a solution map
    $\vG:D\to X$, continuous on
    $D:=\{(\vx_0,t)\in X\times[0,\infty):t<\tau(\vx_0)\}$, with
    $\tau:X\to(0,\infty]$ lower semicontinuous, such that
    for each $\vx_0\in X$ the trajectory $t\mapsto \vG(\vx_0,t)$ is absolutely
    continuous on compact subintervals of $[0,\tau(\vx_0))$, satisfies
    $\vG(\vx_0,0)=\vx_0$, and
    \[
    \frac{d}{dt}\vG(\vx_0,t) \;=\; \QTC\!\big(\vG(\vx_0,t)\big)
    \quad\text{a.e.},
    \]
    and is maximal: no solution in $X$ with initial condition $\vx_0$ extends it
    beyond $\tau(\vx_0)$, and any other such solution agrees on its common
    domain.
\end{assumption}
The following is a sufficient condition for this to hold:
\begin{restatable}{lemma}{lemLipschitz}\label{lemma:Lipschitz}
    If $X$ is closed and convex (Assumption~\ref{ass:convex-closed}) and $\vQ$ is locally Lipschitz on $X$, \ref{ass:well-behaved-proj} is satisfied. If $\vQ$ is globally Lipschitz on $X$, then the solution map is global: $\tau(\vx_0)=\infty$ for every $\vx_0\in X$.
\end{restatable}
\begin{proof}
    The global Lipschitz statement is implied by Theorem 2.5 of
    \citet{nagurney2012projected}. The locally Lipschitz case follows by
    localizing the vector field and patching the resulting global flows; see
    Appendix~\ref{proof:lipschitz-local-flow}.
\end{proof}
\begin{definition}[Dynamic solver]
Given a continuous vector field $\vQ$ whose projected dynamics satisfy \ref{ass:well-behaved-proj}, let $\vG$ be the associated maximal solution map and $\tau$ its maximal existence time. A dynamic solver is $g_{\vQ}:X\to \tilde X$ defined by
\[
    g_{\vQ}(\vx_0)\;:=\;
    \begin{cases}
        \displaystyle \lim_{t\to\infty} \vG(\vx_0,t)
        & \substack{\text{if }\tau(\vx_0)=\infty\\
        \text{and the limit exists in }X,}\\
        \dagger & \text{otherwise}
    \end{cases}
\]
\end{definition}

\begin{definition}[Reachable terminal outcomes (RTO) and $\mu$-observable RTO]\label{def:RTO}
Let $g:X\to\tilde X$ be any solver map. We call the image $g(X)\subseteq \tilde X$ the set of \emph{reachable terminal outcomes (RTO)} (i.e., possible long-run outcomes) induced by $g$.

For $\vy\in \tilde X$, define its basin of attraction as $B_{\vy}^g := g^{-1}(\vy)$ and its size as $s_{\vy}^g := \mu(B_{\vy}^g)=\mu(g^{-1}(\vy))$. We say that an RTO $\vy$ is \emph{$\mu$-observable} if $s_{\vy}^g>0$ (in particular, for a dynamic solver $g_{\vQ}$, this means $s_{\vy}^{g_{\vQ}}=\mu(g_{\vQ}^{-1}(\vy))>0$). An RTO need not be $\mu$-observable: it can exist as a limiting outcome but never appear under random restarts drawn from $\mu$. When $g$ is clear from context, we drop the superscript and write $B_{\vy}$ and $s_{\vy}$.
\end{definition}

\begin{restatable}{proposition}{propMeasurable}\label{prop:measurable}
    Dynamic solvers are measurable.
\end{restatable}
\begin{proofsketch}
Lower semicontinuity of $\tau$ makes each survival set
$A_q=\{\vx:\tau(\vx)>q\}$ Borel. On $A_q$, the time-$q$ flow is continuous;
extending it by a dummy point outside $A_q$ gives a Borel map on $X$. On
$A_\infty=\{\vx:\tau(\vx)=\infty\}$, continuity in $t$ reduces $\limsup$ and
$\liminf$ as $t\to\infty$ to countable suprema and infima over $q\in\Q$. Hence
the convergence set is measurable, and defining $g_{\vQ}=\dagger$ elsewhere
yields a measurable map into $\R^{d}\sqcup\{\dagger\}$. See
Appendix~\ref{proof:measurable}.
\end{proofsketch}

Thus measurability of $g$ follows from well-posed projected dynamics; local Lipschitzness of $\vQ$ is a convenient sufficient condition.

\begin{remark}
    Let $\vQ(\vx) = -\BFnabla f(\vx)$ with $f:X \to \R$ having Lipschitz gradient. Then $g_{\vQ}$ gives us a dynamic solver that corresponds to projected gradient flows (idealized gradient descent) on $f$ constrained to $X$.
\end{remark}

\begin{remark}\label{picard-gq}
    Let $\vQ(\vx) = \vF(\vx) - \vx$ for some Lipschitz operator $\vF:X \to X$. Then $g_{\vQ}$ gives us a dynamic solver that corresponds to projected Picard flows (idealized fixed point iterations) on $\vF$ constrained to $X$.
\end{remark}

\subsection{Finitely Many Reachable Terminal Outcomes}
Now that we have addressed measurability, we turn to finiteness of reachable terminal outcomes. As we only know a solver $g_{\vQ}$ through its dynamics $\vQ$, an ideal condition for finiteness of reachable terminal outcomes should be stated in terms of properties of $\vQ$.

Let $\mR(f)$ be the roots of $f$ on its domain. The following proposition gives a necessary and sufficient condition for finiteness of reachable terminal outcomes in terms of roots of the projected dynamics.
\begin{proposition}\label{prop:finite-rto-roots}
    A dynamic solver $g_{\vQ}$ has finitely many reachable terminal outcomes ($|g_{\vQ}(X)|<\infty$) if and only if $\QTC$ has finitely many roots.
\end{proposition}
\begin{proofsketch}
The nontrivial direction is that every convergent projected trajectory must converge to a root of $\QTC$. This resembles standard results such as Barbalat's lemma, but those results do not apply directly because $\QTC$ need not be continuous at the boundary. If $\vx(t)\to\vx^*$, then for every $\vy\in X$ and for a.e.\ $t$,
\[
\langle \vQ(\vx(t)),\vy-\vx(t)\rangle
\le -\frac12\frac{d}{dt}\|\vy-\vx(t)\|^2.
\]
Averaging this inequality and letting $t\to\infty$ gives
$\langle \vQ(\vx^*),\vy-\vx^*\rangle\le 0$ for all $\vy\in X$, hence
$\vQ(\vx^*)\in N_X(\vx^*)$, which is equivalent to $\QTC(\vx^*)=0$.
Therefore $g_{\vQ}(X)\subseteq\mR(\QTC)\cup\{\dagger\}$. See
Appendix~\ref{proof:containment} for the full proof.
Conversely, each root $\vx^*$ of $\QTC$ is an RTO because the constant trajectory $\vG^*(\vx^*,t)=\vx^*$ solves the projected autonomous ODE with initial condition $\vx^*$.
\end{proofsketch}

For example, when $X=\R^d$ and $\vQ=-\BFnabla f$, the RTOs of the gradient flow are finite iff $f$ has finitely many critical points.

\subsection{Prevalence of Finitely Many Reachable Terminal Outcomes}

Consider a complete metric vector space $\mathcal Q$ of $C^1$ dynamics on $X$,
and suppose each $\vQ\in\mathcal Q$ induces a dynamic solver $g_{\vQ}$. The
corresponding class of dynamic solvers is
\[
\mathcal G_{\mathcal Q}:=\{g_{\vQ}:\vQ\in\mathcal Q\}.
\]
Whenever we use such a family as a class of dynamic solvers, we assume the
regularity conditions from Section~\ref{sec:model} needed for each
$\vQ\in\mathcal Q$ to induce a well-defined solver $g_{\vQ}$. In particular, by
Lemma~\ref{lemma:Lipschitz}, this is automatic when $X$ is closed and convex and
the dynamics in $\mathcal Q$ are locally Lipschitz. If the dynamics are
Lipschitz, the corresponding solution maps are global; on compact domains,
$C^1$ dynamics are Lipschitz.
\begin{remark}
Examples of classes of dynamic solvers include projected gradient flows of
$C^2$ functions by setting $\mathcal{Q}=\{-\BFnabla f:f\in C^2(X,\R)\}$ and
projected Picard flows of $C^1$ functions by setting
$\mathcal{Q}=\{\vx\mapsto \vF(\vx)-\vx:\vF\in C^1(X,X)\}$.
\end{remark}

For a fixed dynamics $\vQ$, finiteness of $g_{\vQ}(X)$ can sometimes be checked
directly by analyzing the roots of $\QTC$. In many applications, however, this
root structure is difficult to verify for a particular $\vQ$. We therefore ask
for a generic statement: conditions under which almost all dynamics in
$\mathcal Q$, in a sense made precise below, induce finitely many reachable
terminal outcomes.

To preview the proof mechanism, suppose $F:\mathbb R^d\to\mathbb R^d$ is a $C^1$ function whose roots all lie in a compact set $K\subseteq \mathbb R^d$. Since $F$ is continuous, its root set $\mathcal R(F)=F^{-1}({0})$ is closed in $\mathbb R^d$; because $\mathcal R(F)\subseteq K$, it is also closed in $K$. If $F\pitchfork 0$ ($F$ is transverse to $0$), then $\mathcal R(F)$ is a zero-dimensional submanifold of $\mathbb R^d$, hence discrete. Thus $\mathcal R(F)$ is closed and discrete inside the compact set $K$, and therefore $\mathcal R(F)$ is finite.

Applying this mechanism to $\QTC$ has two separate hurdles. First, the relevant
root sets must be closed inside compact sets. Second, they must be discrete. The
first hurdle is not automatic because $\QTC$ is not globally continuous or
$C^1$. We therefore reduce its roots to two better-behaved pieces:
\[
\mR(\QTC)\subseteq \mR(\vQ)\cup \mR(\QTCP).
\]
The map $\vQ$ is continuous by definition. If $X$ has $C^2$ boundary, then the
outward normal field $\vn:\partial X\to\R^d$ is $C^1$, so
\[
\QTCP(\vx)=\vQ(\vx)-\langle \vQ(\vx),\vn(\vx)\rangle_+\vn(\vx)
\]
is continuous on $\partial X$. Thus $\mR(\vQ)$ and $\mR(\QTCP)$ are closed in
their respective domains. What remains for the closed-in-compact part is compact
containment of these root sets.

\begin{assumption}[Root compactness]\label{ass:kkq}
\begin{enumerate}
    \item $X$ has $C^2$ boundary.
    \item For every $\vQ\in\mathcal Q$, there exist compact sets
    $K_1(\vQ)\subseteq X$ and $K_2(\vQ)\subseteq \partial X$ such that
    \[
    \mR(\vQ)\subseteq K_1(\vQ),
    \qquad
    \mR(\QTCP)\subseteq K_2(\vQ).
    \]
\end{enumerate}
\end{assumption}

\begin{remark}[A simple sufficient condition]\label{rem:compact-x-sufficient}
The second part of Assumption~\ref{ass:kkq} is automatically satisfied if $X$
is compact. In that case, one may take $K_1(\vQ)=X$ and
$K_2(\vQ)=\partial X$ for every $\vQ\in\mathcal Q$. Thus compactness of the
whole state space is a convenient sufficient condition for root containment, but
the argument only uses compactness to contain the relevant roots.
\end{remark}

When $X$ is not compact, Assumption~\ref{ass:kkq} can sometimes be verified
by problem-specific bounds. For example, in pricing models, the natural price
space is often the positive orthant, which is not compact. If the relevant
equilibria are known to be bounded above and below, then the corresponding roots
lie in a compact region even though the ambient state space is noncompact.

We emphasize that Assumption~\ref{ass:kkq} is used only for results of this
subsection: prevalence-based finiteness of RTOs. If finiteness of
$g_{\vQ}(X)$, equivalently finiteness of roots of $\QTC$, is established by
other means in a given application, Assumption~\ref{ass:kkq} is not needed for
the inferential results in Section~\ref{sec:inference}.

Before moving on to the discreteness requirement, we make the notion of ``almost
all'' precise. When $\mathcal{Q}$ is finite-dimensional, the standard notion of
almost all is full Lebesgue measure. In infinite-dimensional spaces, Lebesgue
measure is not informative, as it is either zero or infinite on all open sets.
Prevalence \citep{hunt1992prevalence} and Haar-null sets
\citep{christensen1972sets} are two generalizations of almost all that agree
with full Lebesgue measure in finite dimensions and still apply in infinite
dimensions. We use prevalence because it is directly adapted to vector spaces.

\begin{definition}
    A subset $S$ of a complete metric vector space $\mathcal{Q}$ is prevalent if there exists a finite-dimensional subspace $P \subseteq \mathcal{Q}$ such that for every $\vQ \in \mathcal{Q}$, the set $\{\vp \in P : \vQ + \vp \in S\}$ has full Lebesgue measure in $P$.
\end{definition}
\begin{remark}
    The standard interpretation of this definition is to take $P$ to be a finite-dimensional vector space of perturbations. $S$ is prevalent if it is stable under perturbations: for any $\vQ \in S$, perturbed $\vQ$ stays in $S$ for almost all perturbations (in the Lebesgue sense).
\end{remark}

We want sufficient conditions on $\mathcal{Q}$ such that, for prevalent
$\vQ\in\mathcal Q$, the dynamic solver has finitely many reachable terminal
outcomes: $|g_{\vQ}(X)|<\infty$

Having handled the closed-in-compact part, the remaining hurdle is discreteness.
We derive structural properties of $\mathcal Q$ that guarantee discreteness of
$\mR(\vQ)$ and $\mR(\QTCP)$ for prevalent $\vQ$.

\begin{definition}
    For a finite-dimensional subspace $P$ of $\mathcal{Q}$, for any $\vQ \in \mathcal{Q}$, define the following perturbed evaluation maps $e^{\vQ}:P\times X\to\R^{d}$ and $e^{\tilde{\vQ}}:P \times \partial X \times [0,\infty) \to \R^{d}$:
    \[
    \begin{aligned}[t]
        e^{\vQ}(\vp,\vx)&:=(\vQ+\vp)(\vx)\\
        e^{\tilde{\vQ}}(\vp,\vx,\alpha)&:=(\vQ+\vp)(\vx) - \alpha \vn(\vx)
    \end{aligned}
    \]
\end{definition}

These perturbed evaluation maps encode the perturbation directions used to
obtain transversality for the two root problems.

\begin{restatable}{proposition}{propESubmersionFiniteAttractors}\label{prop:e-submersion-finite-attractors}
Assume $X$ satisfies Assumption~\ref{ass:convex-closed}, and let $\mathcal Q$ be a complete metric vector space of continuously differentiable functions from $X$ to $\R^d$ satisfying Assumption~\ref{ass:kkq}. Suppose there exists a finite-dimensional subspace $P\subseteq\mathcal Q$ such that, for every $\vQ\in\mathcal Q$, the perturbed evaluation maps $e^{\vQ}:P\times X\to\R^d$ and $e^{\tilde{\vQ}}:P\times\partial X\times[0,\infty)\to\R^d$ are $C^1$ submersions. Then for prevalent $\vQ \in \mathcal{Q}$, there are finitely many RTOs: $|g_{\vQ}(X)| < \infty$.
\end{restatable}
\begin{proofsketch}
Lemma~\ref{lemma:containment} reduces finiteness of reachable terminal outcomes to
finiteness of roots of $\QTC$. Appendix~\ref{appendix:prevalence} then shows
that
\[
\mR(\QTC)\subseteq \mR(\vQ)\cup \mR(\QTCP),
\]
where $\mR(\vQ)$ and $\mR(\QTCP)$ are closed in their respective domains under
Assumption~\ref{ass:kkq}. The submersion hypotheses on $e^{\vQ}$ and
$e^{\tilde{\vQ}}$ are exactly the hypotheses used in the appendix's parametric
transversality argument: Sard's theorem applied to the zero set of the
perturbed evaluation map implies that almost every perturbation makes the fiber
map transverse to zero. Hence, for prevalent $\vQ\in\mathcal Q$, $\mR(\vQ)$ is
discrete in $X$ and $\mR(\QTCP)$ is discrete in $\partial X$. Assumption~\ref{ass:kkq}
places these closed discrete root sets inside compact sets, so they are finite.
The containment above then implies $|g_{\vQ}(X)|<\infty$.
Appendix~\ref{proof:e-submersion-finite-attractors} provides a formal proof.
\end{proofsketch}

As the evaluation-map submersion requirements in Theorem~\ref{theorem:prevalence_submersion_probe} can feel intangible, we provide the following sufficient condition that is easier to grasp. This condition will be used in Section~\ref{sec:applications} to verify prevalence of finitely many RTO.

\begin{proposition}\label{prop:ev-submersion}
Take $P \subseteq \mathcal{Q}$ to be a finite-dimensional subspace such that for every $\vx\in X$,
the evaluation map $\ev_{\vx}:P\to\R^d$ is a submersion.
Then for every $\vQ\in\mathcal{Q}$, the perturbed evaluation maps $e^{\vQ},e^{\tilde{\vQ}}$ are $C^1$ submersions.
\end{proposition}
\begin{proof}
We have already established that $e^{\vQ}$ and $e^{\tilde{\vQ}}$ are $C^1$. It remains to show they are submersions.
Fix $\vQ\in\mathcal Q$ and $(\vp,\vx)\in P\times X$.

Since $P$ is a finite-dimensional vector space, we identify $T_{\vp}P\simeq P$ by translation.
For $(\dot \vp,\dot \vx)\in T_{\vp}P\times T_{\vx}X\simeq P\times T_{\vx}X$,
\[
D e^{\vQ}_{(\vp,\vx)}(\dot \vp,\dot \vx)=\dot \vp(\vx)+D(\vQ+\vp)_{\vx}(\dot \vx).
\]
In particular, restricting to parameter directions,
\[
D e^{\vQ}_{(\vp,\vx)}(\dot \vp,0)=\dot \vp(\vx)=D(\ev_{\vx})_{\vp}(\dot \vp).
\]
By assumption $D(\ev_{\vx})_{\vp}:T_{\vp}P\to\R^d$ is surjective, hence $D e^{\vQ}_{(\vp,\vx)}$ is surjective. Therefore $e^{\vQ}$ is a submersion. The same argument applies to $e^{\tilde{\vQ}}$ since the term $-\alpha \vn(\vx)$ does not depend on $\vp$.
\end{proof}

\begin{remark}\label{remark:finite-picard-gradient}
    Let $\mathcal{C}$ be the class of constant scalar-valued functions:
    $\mathcal{C}=\{\vF \mid \vF:X \to \R^d \land \exists \vu \in \R^d \;\forall \vx \in X:\; \vF(\vx)=\vu\}$.
    On $\mathcal{C}$, for all $\vx \in X$, the evaluation map $\ev_{\vx}:\mathcal{C}\to\R^d$ is a diffeomorphism hence a submersion. Therefore, if $\mathcal{C} \subseteq \mathcal{Q}$,  Proposition~\ref{prop:ev-submersion} applies with $P=\mathcal{C}$ and the conditions of Theorem~\ref{theorem:prevalence_submersion_probe} are satisfied. For gradient descent of $C^2$ functions, under Assumption~\ref{ass:convex-closed}, $\mathcal{Q}=\{\BFnabla f|f \in C^2(X, \R^d)\}$ makes a complete metric vector space containing $\mathcal{C}$. For Picard flows of $C^1$ functions, $\mathcal{Q}=\{\vx \mapsto \vF(\vx) - \vx \mid \vF \in C^1(X, \R^d)\}=C^1(X, \R^d)$ also makes a complete metric vector space containing $\mathcal{C}$. Hence for prevalent gradient flows and Picard flows, there are finitely many reachable terminal outcomes.
\end{remark}

\subsection{Numerical Inaccuracies and Observed Partitions}
\label{subsec:numerical}
Solvers in practice may not produce or record the exact output of $g_{\vQ}$ due to numerical inaccuracies. Moreover, the notion of ``same solution'' is operational: two outputs are identified if they are indistinguishable at a prescribed numerical tolerance. To model this, we introduce a measurable \emph{observation map}
$\mathcal O_\varepsilon:\tilde X \to \tilde X$ where $\varepsilon>0$ represents the effective numerical resolution (e.g., floating-point precision,
stopping criteria, etc). The observed terminal map is $g_{\vQ}^{\varepsilon}:=\mathcal O_\varepsilon\circ g_{\vQ}$.

It induces the \emph{observed partition} of initial conditions into basins of observed terminal outcomes. Decreasing $\varepsilon$ may refine this observed partition by splitting sets of initial conditions whose terminal points agreed at coarser resolution.

With measurable $g_{\vQ}$ that has finitely many reachable terminal outcomes, $g_{\vQ}^{\varepsilon}$
is also measurable and has finitely many reachable terminal outcomes. Hence the analysis of this paper
applies verbatim to $g_{\vQ}^{\varepsilon}$.

The statistical analysis of the next section depends only on the partition of $X$ into basins induced
by $g_{\vQ}^{\varepsilon}$. Thus $\mathcal O_\varepsilon$ matters only insofar as it changes
this partition. Varying $\varepsilon$ corresponds to coarsenings/refinements of the observed partition,
which is exactly the setting of Assumption~\ref{ass:refine_coarsen}. In practice, small numerical errors do not significantly change the induced partition on $X$. Hence we work with $\mathcal O_\varepsilon=\mathrm{Id}$.
When numerical imprecision changes the partition, the results of the next section should be read as inference about the
observed partition induced by $g_{\vQ}^{\varepsilon}$ rather than the idealized one induced by $g_{\vQ}$.

Notice that the finiteness of numerically recorded outcomes should not be treated as a substitute for the RTO finiteness conditions of the previous sections. If the idealized problem has infinitely many terminal outcomes, then the observed partition induced by $g_{\vQ}^{\varepsilon}$ will not coincide with the partition induced by $g_{\vQ}$, which is our object of interest. Section~\ref{sec:limits} discusses how numerical inaccuracies affect the interpretation of restart evidence.

\section{Inference}\label{sec:inference}
Once the categorical reduction is in place, repeated convergence to a single RTO supports two distinct inferential tasks. The first is \emph{basin-size inference}: how much restart mass is assigned to the observed RTO. The second is \emph{multiplicity inference}: whether other RTOs can also occur under random restarts.

The two tasks have different statistical difficulty. Repeated convergence directly supports the claim that the observed RTO has a large basin under the restart distribution. By contrast, uniqueness or multiplicity claims require ruling out other RTOs whose basins may be very small. The results and rates derived below formalize this distinction.

Let $\vx_1, \vx_2, \dots, \vx_n$ be i.i.d.\ samples from $\mu$ and define $\vy_i := g_{\vQ}(\vx_i)$ for some dynamic solver $g_{\vQ}$. Define the event of getting the same RTO $\vx^* \in \tilde{X}$ for all $n$ runs as:
\[
H_n(\vx^*) = \{\vy_1 = \vy_2 = \dots = \vy_n=\vx^*\}.
\]

We show how conditioning on $H_n$ changes our beliefs about the size of $\vx^*$'s basin of attraction $s_{\vx^*}$ and the number of RTO. We start with basin-size inference and then turn to the harder problem of multiplicity inference.

\subsection{Basins of Attraction}
Given a prior $\Pi$ on $s_{\vx^*}$, we are interested in computing the posterior:
\[
    \Pi(s_{\vx^*} \mid H_n(\vx^*)) = \frac{\PP(H_n(\vx^*) \mid s_{\vx^*}) \Pi(s_{\vx^*})}{\PP(H_n(\vx^*))}.
\]
Note that the likelihood $\PP(H_n(\vx^*) \mid s_{\vx^*}) = s_{\vx^*}^n$ as each independent run has probability $s_{\vx^*}$ of converging to $\vx^*$.

We are interested in how this posterior concentrates around 1 as $n$ increases. Dually, that is how $\Pi(A_\varepsilon\mid H_n(\vx^*))$ concentrates around $0$ where $A_\varepsilon = \{s_{\vx^*} \le 1 - \varepsilon\}$ for some $\varepsilon > 0$.

A simple choice is to use a Beta prior $\operatorname{Beta}(\alpha, \beta)$ for some $\alpha, \beta > 0$. This prior has full support on $[0,1]$ and is conjugate to the Bernoulli likelihood that will arise from our observations. Appendix~\ref{proof:beta_prior} proves the following result using standard properties of the Beta distribution.

\begin{restatable}{proposition}{propBetaPrior}\label{prop:beta_prior}
    Under a $\operatorname{Beta}(\alpha, \beta)$ prior on $s_{\vx^*}$, the posterior after observing $H_n(\vx^*)$ is $\operatorname{Beta}(\alpha + n, \beta)$. Additionally, the posterior concentrates around 1 exponentially fast in $n$:
    \begin{equation}\label{eq:basin_tail}
    \Pi\!\left(A_\varepsilon \mid H_n(\vx^*)\right)
    \;\le\;
    \frac{(\alpha+n+\beta)^{\beta}}{\Gamma(\beta+1)}\,
    (1-\varepsilon)^{\alpha+n-1}
    \;
\end{equation}
\end{restatable}

This result formalizes the intuition that observing repeated convergence to the same RTO increases our belief that its basin of attraction is large. The exponential concentration rate shows that this increase in belief is quite rapid, hence hundreds of repeated runs can provide strong evidence for a dominant basin of attraction.

Though this result is derived under a Beta prior, we generalize this to a wide class of priors. Essentially, we need to impose some regularity conditions on the prior near 1, ensuring that the prior does not vanish as we approach 1. This yields the following generalization of the exponential concentration result:
\begin{restatable}{proposition}{propEtaBound}\label{prop:eta-bound}
    If $1 \in \operatorname{supp}(\Pi)$, then for any $0 < \varepsilon < 1$, $\Pi(A_\varepsilon \mid H_n(\vx^*))$ converges to 0 exponentially fast in $n$. In particular, there exists some $0 < \eta < \varepsilon$ such that
    \[
        \Pi(A_\varepsilon \mid H_n(\vx^*))
        \leq
        \frac{(1-\varepsilon)^n\Pi(A_\varepsilon)}
        {(1-\eta)^n(1-\Pi(A_\eta))}
    \]
\end{restatable}
\begin{proofsketch}
The proof uses two elementary facts, proved in Appendix~\ref{proof:eta-bound}. First, Bayes' rule gives the following bound by upper bounding $\PP(H_n(\vx^*)\cap A_\varepsilon)$ and lower bounding $\PP(H_n(\vx^*))$ on $A_\eta^c$: for any $0<\eta<\varepsilon<1$ with $\Pi(A_\eta)<1$,
\[
\Pi(A_\varepsilon\mid H_n(\vx^*))
\le
\frac{(1-\varepsilon)^n\Pi(A_\varepsilon)}
{(1-\eta)^n(1-\Pi(A_\eta))}.
\]
Second, $1\in\operatorname{supp}(\Pi)$ guarantees the existence of such an $\eta$. Combining these facts gives exponential decay because $(1-\varepsilon)/(1-\eta)<1$.
\end{proofsketch}

So under the mild condition that the prior has support arbitrarily close to 1, we again obtain exponential concentration bounds of the posterior around 1. That result only needs support near 1 to choose some $\eta<\epsilon$.
If we want an explicit bound independent of $\eta$, we need more: the lower bound on the prior density near 1 from Assumption~\ref{ass:poly_density} is exactly what produces the polynomial prefactor in the next result.

\begin{assumption}[Polynomial lower bound of prior density near $1$]\label{ass:poly_density}
The prior $\Pi$ on $s_{\vx^*}\in[0,1]$ admits a density $f_s$ and there exist constants
$c>0$, $\delta\in(0,1)$, and $\kappa\ge 0$ such that
\[
f_s(s)\ge c(1-s)^{\kappa}\qquad \forall\, s\in(1-\delta,1).
\]
\end{assumption}

\begin{restatable}[Dropping $\eta$ under Assumption~\ref{ass:poly_density}]{proposition}{propBasinPosteriorSize}\label{prop:basin-posterior-size}
Under Assumption~\ref{ass:poly_density}, for any $\varepsilon\in(0,1)$ and all $n \geq 2$
large enough so that $1/n<\delta$,
\begin{equation}\label{eq:drop_eta_poly_bound}
\begin{aligned}
\Pi(A_\varepsilon\mid H_n(\vx^*)) &\le \frac{4(\kappa+1)}{c}\,(n)^{\kappa+1}(1-\varepsilon)^{n} \Pi(A_\varepsilon) \\
&\le \frac{4(\kappa+1)}{c}\,(n)^{\kappa+1}e^{-(n)\varepsilon} \Pi(A_\varepsilon).
\end{aligned}
\end{equation}
\end{restatable}

\begin{proofsketch}
Lemma~\ref{lem:eta_bound} yields an upper bound for $\Pi(A_\varepsilon \mid H_n(\vx^*))$ whose denominator $(1-\eta)^n (1-\Pi(A_\eta))$ depends on $\eta$. We choose $\eta$ as a function of $n$ and lower bound $1-\Pi(A_\eta)$ using Assumption~\ref{ass:poly_density} to obtain an explicit (i.e.\ $\eta$-free) bound.
Set $\eta = \frac{1}{n}$ and assume $1/n<\delta$ so that Assumption~\ref{ass:poly_density} applies. For $n\ge 2$, the map $n\mapsto (1-\frac1n)^n$ is increasing, hence $(1-\frac1n)^n\ge (1-\frac12)^2=\frac14$ and equivalently $(1-\frac1n)^{-n}\le 4$. Under Assumption~\ref{ass:poly_density}, for $\eta<\delta$ we have
\[
\begin{aligned}
1-\Pi(A_\eta)=\Pi(s_{\vx^*}>1-\eta) \\
\ge \int_{1-\eta}^1 c(1-s)^\kappa\,ds=\frac{c}{\kappa+1}\eta^{\kappa+1}
\end{aligned}
\]
With $\eta=1/n$, this gives $(1-\Pi(A_{1/n}))^{-1}\le \frac{\kappa+1}{c}n^{\kappa+1}$, and combining with the previous bounds yields \eqref{eq:drop_eta_poly_bound}.
Appendix~\ref{proof:basin-posterior-size} provides a detailed proof.
\end{proofsketch}

\subsection{Number of Reachable Terminal Outcomes}
A necessary condition for uniqueness of an RTO is that its basin of attraction has size one. Yet under any prior $s_{\vx}$ that admits a density, this event has zero prior probability. The reason is that under the prior, the event $s_{\vx}=1$ has probability zero, and Bayesian updating will not give a positive probability to an event that has zero prior probability. So to talk about uniqueness we need more. We introduce two approaches: a simpler spike and slab prior and a heavier Mixture of Finite Models prior.

\subsubsection{Spike-and-Slab Prior}
So far, we have focused on the sizes of basins of attraction. Definition~\ref{def:RTO} gives the relevant visibility notion: an RTO is $\mu$-observable exactly when its basin has positive restart mass. If all RTO are $\mu$-observable, then having a unique RTO is equivalent to the observed RTO having basin size one.

In practice, however, this assumption may be too strong. For example, we may want to know whether the local minima of a Morse function $f$ found by GD are global. The Morse condition rules out flatness, hence all local minima have basins of positive size. The problem is that saddle points are also RTO but they have basins of size zero.

To overcome this, we reformulate our question about uniqueness. Instead of asking whether there is a unique RTO, we ask whether there is a unique $\mu$-observable RTO. In the GD example, this corresponds to asking whether there is a unique local minimum rather than whether there is a unique RTO.

\begin{lemma}
	    $\vx^* \in \tilde{X}$ is the unique $\mu$-observable RTO iff its basin of attraction has size one.
\end{lemma}
\begin{proof}
If $\vx^*$ is the unique $\mu$-observable RTO, then $s_{\vx}=0$ for all $\vx\in\tilde X\setminus\{\vx^*\}$, hence
$s_{\vx^*}=1$.
Conversely, if $s_{\vx^*}=1$, then $s_{\vx}=0$ for all $\vx\in\tilde X\setminus\{\vx^*\}$, so no other RTO is
$\mu$-observable.
\end{proof}

Now we can go back to our previous analysis about basin sizes.

Remember that any event that has probability zero under the prior will have probability zero under the posterior. So any prior over a basin size that admits density such as the ones discussed so far will never give a positive probability for uniqueness of a $\mu$-observable RTO. Hence we need a prior that assigns positive probability to the event $s_{\vx^*} = 1$.

Spike and slab priors are a popular choice for such situations where we want to assign positive probability to a specific value (the spike) while allowing for uncertainty around other values (the slab). They are a mixture of a point mass at a specific value and a continuous distribution over the remaining values.

Let $\delta_{x}$ be the Dirac delta measure centered at $x$, and let $W$ be a measure over $[0,1]$ that admits density. Define a random variable $Z$ that determines whether the spike is true or the slab. The following is our spike and slab prior:
\[
    s_{\vx^*} | Z=0 \sim W \quad , \quad
    s_{\vx^*} | Z=1 = 1
\]
If $\Pi(Z=1) = p$, then the marginal prior on $s_{\vx^*}$ is $\Pi = p \delta_{1} + (1-p) W$. Throughout this subsection we assume $p>0$.

Now we are interested in $\Pi(s_{\vx^*}=1\mid H_n(\vx^*))$. Notice that $\Pi(s_{\vx^*} = 1\mid H_n(\vx^*)) = \Pi(Z=1\mid H_n(\vx^*))$. Using Bayes' rule, we have:
\begin{align}
    \Pi(s_{\vx^*} &= 1\mid H_n(\vx^*)) = \Pi(Z=1\mid H_n(\vx^*)) \notag \\
    &= \frac{p}{p + (1-p)\,\EE_W[s_{\vx^*}^{n}]}\label{eq:z=1_posterior}
\end{align}

As $W$ admits density over $[0,1]$, $\EE_W[s_{\vx^*}^{n}] \to 0$ as
$n \to \infty$. Hence \eqref{eq:z=1_posterior} implies that the posterior
probability of uniqueness of a $\mu$-observable RTO tends to $1$.

To talk about the rate of convergence, we impose a tail condition on the slab near $1$.

\begin{restatable}[Spike-and-slab posterior rate]{proposition}{propSpikeSlabRate}\label{prop:spike_slab_rate}
Suppose the slab $W$ satisfies
\[
W([1-u,1]) \le C u^\gamma \qquad \forall\,u\in(0,\delta]
\]
for some $C>0$, $\gamma>0$, and $\delta>0$. Then, under the spike-and-slab prior defined above,
\[
1-\Pi(s_{\vx^*}=1\mid H_n(\vx^*))
=O(n^{-\gamma}).
\]
\end{restatable}
\begin{proofsketch}
Let $S=s_{\vx^*}$ under the slab $W$. By the layer-cake representation,
\[
\begin{aligned}
\EE_W[S^n]
&=
n\int_0^\delta (1-u)^{n-1}W(S\ge 1-u)\,du \\
&\qquad{}
+n\int_\delta^1 (1-u)^{n-1}W(S\ge 1-u)\,du .
\end{aligned}
\]
The first term is $O(n^{-\gamma})$ by the tail assumption
$W(S\ge 1-u)\le Cu^\gamma$ and the Gamma-integral bound. The second term is
exponentially small because $W(S\ge 1-u)\le1$ and
$(1-u)^{n-1}\le(1-\delta)^{n-1}$ for $u\ge\delta$. Thus
$\EE_W[S^n]=O(n^{-\gamma})$, and substituting this into
\eqref{eq:z=1_posterior} gives the result. See
Appendix~\ref{proof:spike_slab_rate} for details.
\end{proofsketch}

\subsubsection{Mixture of Finite Models}
The spike and slab prior is simple and intuitive. However, it only lets us test whether there is a unique RTO. One way to extend this is to model $(\vy_1,\dots,\vy_n)$ as i.i.d.\ draws from a categorical distribution whose number of components is unknown. This is a Mixture of Finite Models (MFM) \citep{miller2018mixture} specialized to categorical data (component labels observed), with no component-specific parameters.

MFMs are a class of Bayesian models that allow for uncertainty in the number of components in a mixture model. They are particularly useful when the true number of components is unknown and we want to infer it from the data.

In our context, an MFM constitutes a prior over the number of components $\PP(K)$ and a distribution of basin sizes given $K$:
\[
K\sim P_K,\qquad
s\mid K=k\sim \operatorname{Dirichlet}(\alpha_k^1,\dots,\alpha_k^k).
\]

The model is still underspecified: both $P_K$ and the Dirichlet parameters
$\alpha_k^i$ must be chosen. Exchangeability of the RTO labels suggests the
symmetric choice $\alpha_k^i=\alpha_k$. To tie the remaining sequence
$(\alpha_k)_{k\ge1}$ together, we impose the following refinement--coarsening
invariance condition.

\begin{assumption}[Refinement--coarsening invariance]\label{ass:refine_coarsen}
For any $K\ge1$ and $m\ge2$, let
$\vbasin^{mK}=(s_{j,\ell})_{1\le j\le K,\ 1\le \ell\le m}\sim \Pi_{mK}$ be the basin-mass vector
for $mK$ RTO. Define the coarsening map
\[
\mathcal C_{K,m}(\vbasin^{mK}) \;:=\; \Bigl(\sum_{\ell=1}^m s_{1,\ell},\ \dots,\ \sum_{\ell=1}^m s_{K,\ell}\Bigr)\in\Delta_{K-1}
\]
We assume $\mathcal C_{K,m}(\vbasin^{mK})\sim \Pi_K$ for all $K,m$.
\end{assumption}
This mirrors the observed-partition logic of Section~\ref{subsec:numerical}: an observation map $\mathcal O_\varepsilon$ can coarsen the basins by replacing $g_{\vQ}$ by $g_{\vQ}^{\varepsilon}=\mathcal O_\varepsilon\circ g_{\vQ}$. In our setting, we impose invariance under \emph{refinement--coarsening}: if each RTO is split into $m$ indistinguishable sub-RTO and we then coarsen by summing the $m$ sub-basins, the induced law on the coarsened basin sizes should agree with the $K$-RTO prior.
This is the assumption that lets the Dirichlet concentration parameter be tied across different values of $K$; without it, the family $(\alpha_K)_{K\ge1}$ would remain underdetermined.

\begin{proposition}[Coarsening invariance pins $\alpha_k$]
Assume that for each $K\ge1$ the prior on basin masses is symmetric Dirichlet. Then under Assumption~\ref{ass:refine_coarsen}, $a_K=\alpha/K$ for some constant $\alpha>0$.
\end{proposition}

\begin{proof}
By the aggregation property of the Dirichlet distribution, if $(s_{j,\ell})\sim
\operatorname{Dirichlet}(a_{mK},\dots,a_{mK})$, then the block sums
$S_j=\sum_{\ell=1}^m s_{j,\ell}$ satisfy $(S_1,\dots,S_K)\sim
\operatorname{Dirichlet}(m a_{mK},\dots,m a_{mK})$.
Refinement--coarsening invariance therefore implies $a_K=m a_{mK}$ for all $m,K$.
Taking $K=1$ gives $a_m=a_1/m$, hence $a_K=a_1/K$. Writing $\alpha:=a_1$ yields
$a_K=\alpha/K$.
\end{proof}

Thus far, we have reduced infinitely many parameters $\alpha_k^i$ to a single parameter $\alpha$. $\alpha$ captures how concentrated the distribution of basins of attraction is: smaller $\alpha$ gives spikier distributions while larger $\alpha$ gives more uniform distributions. Hence, if $\alpha$ is small, $H_n(\vx^*)$ is not as informative about $K$ compared to when $\alpha$ is large. So our inference is contingent on the choice of $\alpha$.

Since the prior depends on $P_K$ and $\alpha$, we can either adopt conservative defaults
(e.g.\ a heavy-tailed $P_K$ and small $\alpha$), or calibrate these hyperparameters on benchmark instances from the same problem class.

To do so, parametrize $P_K(\cdot;\vtheta)$ by a finite-dimensional parameter $\vtheta$. Sample independent instances $\vQ_i\in\mathcal Q$ from an application-specific instance
distribution, and for each instance run the solver from $m$ i.i.d.\ initial conditions $\vx^i_t\sim\mu$ to obtain terminal outcomes $\vy^i_t := g_{\vQ_i}(\vx^i_t)$, $t=1,\dots,m$. We then choose $(\hat{\vtheta},\hat\alpha)$ by empirical Bayes, i.e.\ by maximizing the marginal likelihood of the observed outcomes $\{(\vy^i_1,\dots,\vy^i_m)\}_{i=1}^R$ under the model (integrating out basin sizes and mixing over $K$).

For each $k\ge 1$ and $n\ge 1$, define $L_k(n)\;:=\;\PP(H_n(\vx^*)\mid K=k)$.
Under the symmetric Dirichlet prior, this quantity does not depend on the
particular label assigned to $\vx^*$. Appendix~\ref{proof:Lk_twosided_fixed}
proves explicit two-sided bounds; in particular, for fixed $k$ and
$\delta_k=\alpha(1-\tfrac1k)$, $L_k(n)\asymp n^{-\delta_k}$.

Using these bounds, we can now derive rate-tight concentration bounds for $\PP(K=1\mid \alpha,H_n(\vx^*))$.

\begin{restatable}[Rate-tight concentration of $\PP(K=1\mid H_n(\vx^*))$]{theorem}{thmKOneTightRateFixed}\label{thm:K1_tight_rate_fixed}
Assume $\pi_1=\PP(K=1)>0$ and $\pi_2=\PP(K=2)>0$ (under the fixed hyperparameter $\alpha>0$
built into $\PP$). Then for all $n\ge 1$,
\begin{equation*}
\begin{aligned}
\frac{\pi_2}{2}\left(\frac{\alpha/2}{\,n-1+\alpha/2\,}\right)^{\alpha/2}
&\le 1-\PP(K=1\mid H_n(\vx^*))\\
&\le \frac{1-\pi_1}{\pi_1} \frac12\left(\frac{1+\alpha}{\,n+\alpha\,}\right)^{\alpha/2}
\end{aligned}
\end{equation*}
Consequently, $1-\PP(K=1\mid H_n(\vx^*))=\Theta\!\left(n^{-\alpha/2}\right)$, hence $\PP(K=1\mid H_n(\vx^*))\to 1$ at rate $n^{-\alpha/2}$.
\end{restatable}
\begin{proofsketch}
    Write $N_n:=\sum_{k\ge 2}\pi_k L_k(n)$. By Lemma~\ref{lem:Lk_monotone_fixed},
    $L_k(n)$ is decreasing in $k$, and $1-\PP(K=1\mid H_n(\vx^*))=\frac{N_n}{\pi_1+N_n}$, so:
    \[
    \begin{aligned}
        \pi_2 L_2(n) &\leq N_n \leq 1-\PP(K=1\mid H_n(\vx^*))\\
        &\leq \frac{N_n}{\pi_1} \leq \frac{1-\pi_1}{\pi_1} L_2(n)
    \end{aligned}
    \]
    We can then apply the bounds on $L_2(n)$ from Lemma~\ref{lem:Lk_twosided_fixed}.
    The key mechanism is that the two-sided bounds in Lemma~\ref{lem:Lk_twosided_fixed}
    scale like $n^{-\delta_k}$ with $\delta_k=\alpha(1-\tfrac1k)$, which is minimized at $k=2$.
    Thus the tail mixture over $k\ge2$ is controlled by the $k=2$ term and yields the rate
    $n^{-\alpha/2}$. See Appendix~\ref{proof:K1_tight_fixed} for details.
\end{proofsketch}

The polynomial rate in Theorem~\ref{thm:K1_tight_rate_fixed} reflects the fact that observing $H_n(\vx^*)$ provides direct evidence that the basin mass of $\vx^*$ is large, but only indirect evidence against the existence of \emph{other} $\mu$-observable RTO with small basin masses.

Under the MFM prior, configurations with $K\ge 2$ can still assign a very large share of mass to a single outcome and spread the remaining mass thinly across additional outcomes. Such alternatives are difficult to rule out from finitely many restarts, so posterior concentration on $K=1$ is necessarily much slower than the exponential concentration obtained for basin-size inference.

\section{Application}\label{sec:applications}

This section illustrates how a practitioner can verify the assumptions of Sections~\ref{sec:model}--\ref{sec:inference} in a concrete restart experiment, after which the posterior calculations in Section~\ref{sec:inference} apply without modification. We demonstrate this through the fixed-point equilibrium solver of \citet{morrow2011fixed} for markets with mixed logit demand and multi-product firms.

This choice is not driven by any intimate connection between discrete-choice models and our methods, but by the model's familiarity, widespread use, and the importance of both uniqueness and dominance of its equilibria. It is also a good application for our method because many runs can be generated in reasonable time: implementations such as \texttt{pyblp} make repeated Morrow--Skerlos fixed-point solves realistic in applied mixed-logit work \citep{conlon2020best}. Our goal is methodological: to show how one takes a problem and solver, identifies the structures and assumptions of Section~\ref{sec:model}, and interprets repeated convergence using Section~\ref{sec:inference}.

\subsection{The Morrow--Skerlos Solver}
Empirical industrial organization provides a useful set of examples because it studies flexible equilibrium models for which tractable uniqueness guarantees are often unavailable. Random-restart and equilibrium-selection concerns arise across static supply and demand \citep{BerryLevinsohnPakes1995}, vertical relations \citep{crawford2012welfare, cussen2026nash}, and dynamic oligopoly \citep{PakesMcGuire1994,besanko2010learning}.

We focus on markets with mixed logit demand and multi-product firms, a workhorse setting in static supply and demand estimation. In this setting, multiple equilibria need not prevent estimation, but equilibrium uniqueness matters for counterfactual analysis, prediction, and the construction of optimal instruments \citep{BerryLevinsohnPakes1995, conlon2020best}. Because uniqueness is rarely guaranteed analytically, applied researchers often rely on iterative solvers and restart-based diagnostics. We use the fixed-point solver of \citet{morrow2011fixed} as a concrete example of how the assumptions in Section~\ref{sec:model} can be verified.

We now set up the problem \citet{morrow2011fixed} solves. Let there be $d$ inside goods and an outside option (with price normalized to $0$) in the market, and assume firms have constant returns to scale. $\vc \in (0,\infty)^d$ denotes the vector of marginal costs and $\vx \in (0,\infty)^{d}$ the vector of prices. The vector field constructed from firms' first-order conditions can be rewritten as
\begin{equation}
\vF(\vx) = \vx - \vc - \vzeta(\vx)\label{eqn:ms_zeta_vf}
\end{equation}
where $\vzeta$ is the vector of markups whose expression is provided by \citet{morrow2011fixed} in terms of market shares without closed form. They propose to solve \eqref{eqn:ms_zeta_vf} by doing fixed-point iterations on $\vx \mapsto \vc + \vzeta(\vx)$. Although the roots of \eqref{eqn:ms_zeta_vf} solve all firms' first-order conditions, they need not be equilibria; the solver therefore generates equilibrium candidates whose validity must be checked separately. Morrow and Skerlos point out that, when multiple solutions may exist, it is important to restart computations from different points in the strategy space. We take this restart experiment as the object to which our inference methods are applied.

\subsection{Verifying the Structural Conditions}

We now verify the structural conditions from Section~\ref{sec:model} for the Morrow--Skerlos solver. In the language of Section~\ref{sec:strategy}, this means justifying the categorical reduction by showing that the solver corresponds to a dynamic solver, that the induced terminal-outcome map is measurable, and that there are finitely many RTOs, at least prevalently in the relevant class of dynamics.

\subsubsection{Dynamic Solver}
We first pin down the space of initial guesses and possible outcomes of the solver. Here, this is naturally $X = \R^d_+$, since the objects of interest are prices and they are non-negative. This $X$ satisfies Assumption~\ref{ass:convex-closed} as it is closed and convex.

We next identify the updating rule and its idealized dynamic version. The solver performs fixed-point iterations on $\vx \mapsto \vc + \vzeta(\vx)$, so its idealized version is the Picard flow generated by
\[
\vQ(\vx) = \vc + \vzeta(\vx) - \vx.
\]
\citet{skerlos2010fixed} show that this map is $C^1$ and consequently locally Lipschitz. Hence it corresponds to a dynamic solver $g_{\vQ}$, which is measurable by Proposition~\ref{prop:measurable}.

\subsubsection{Finiteness}
Here, we would need to verify that $\mR(\QTC)$ is finite. In this case, this seems daunting as we don't even have a closed form expression for $\vzeta$. Instead, we aim to show it holds prevalently over the classes of dynamics. To do so, we need to verify three conditions:
\begin{enumerate}
    \item Assumption~\ref{ass:kkq}, which gives compact containment of roots;
    \item a complete vector-space structure on the parametrized class $\mathcal Q$;
    \item the submersion condition needed for Proposition~\ref{prop:e-submersion-finite-attractors}, or equivalently the simpler sufficient condition in Proposition~\ref{prop:ev-submersion}.
\end{enumerate}
As closed and discrete subsets of compact sets are finite, these conditions imply that prevalently $|g_{\vQ}(X)| < \infty$. Remark~\ref{remark:finite-picard-gradient} gives an easier route when $\vzeta$ can vary freely over $C^1(X,\R^d)$. In that case, the Picard dynamics $\vQ(\vx)=\vc+\vzeta(\vx)-\vx$ also range over $C^1(X,\R^d)$, so $\mathcal Q=C^1(X,\R^d)$ and the remark implies prevalent finiteness of RTOs. We give a more explicit argument here because many applications restrict $\vzeta$ to a parametrized family induced by primitives. Prevalence in the ambient $C^1$ class need not imply prevalence inside such a restricted family, so the argument below verifies the conditions directly for the parametrized class.

\citet{skerlos2010fixed} prove that for a wide class of models, for any fixed primitive, the equilibrium prices are bounded below and above by positive constants. Hence they belong to a compact box. As this box sits strictly inside the positive orthant, we can replace it with a slightly larger smooth compact set; Appendix~\ref{appendix:application-smooth-domain} gives an explicit construction. This verifies Assumption~\ref{ass:kkq}.

Moving to the class of dynamics, the natural choice is one obtained by varying the primitives of the problem, i.e. the marginal costs and the markups. Canonically, marginal costs are parameterized by $\vtheta \in \R^{d_C}$ and markups by $\vomega \in \R^{d_\zeta}$ \citep{conlon2020best}. Under an injective parametrization $(\vtheta,\vomega)\mapsto \vQ_{(\vtheta,\vomega)}$, Appendix~\ref{appendix:application-banach} endows $\mathcal{Q}$ with the Banach structure transported from $\R^{d_C}\times\R^{d_\zeta}$. Therefore $\mathcal{Q}$ is indeed a complete metric vector space. This prevalence statement is relative to the chosen parametrization of primitives and the transported linear structure on $\mathcal Q$.

Consider the perturbations $P=\{\vQ_{(\vtheta,0)}|\vtheta \in \R^{d_C}\}$, that is, perturbations that keep markups as is and move marginal costs. If the structure of $\vc_{\vtheta}$ is rich enough to be a submersion, i.e., we can move the marginal cost in any direction by changing $\vtheta$, then we can invoke Proposition~\ref{prop:ev-submersion} to show evaluation maps are submersions and then Proposition~\ref{prop:e-submersion-finite-attractors}. A straightforward choice is to set $d_C=d$ and $\vc_{\vtheta} = e^{\vtheta}$ where $e^{\vtheta}$ is the element-wise exponential of $\vtheta$. (This happens to be a standard function for marginal costs in mixed logit models \citep{conlon2020best}.) In this case, the Jacobian of $\vc_{\vtheta}$ is a diagonal matrix with non-zero entries, so it is a submersion. Hence, having finitely many reachable terminal outcomes is prevalent in $\mathcal{Q}$.

Having verified the structural conditions, the categorical reduction applies: repeated random restarts can be treated as categorical data over the RTOs induced by $g_{\vQ}$.

\subsection{Illustrative Posterior Calculations}

Once the structural conditions are verified, Tables~\ref{table:basin-posterior}, \ref{table:spike-slab}, and~\ref{table:MFM-bound} show how the posteriors of Section~\ref{sec:inference} concentrate after observing $H_n$, the event that $n$ restarts converge to the same RTO.

\begin{table}[t]
\TABLE
{Posterior probability $\Pi(s_{\vx^*}\ge 1-\varepsilon\mid H_n)$ under $s_{\vx^*}\sim \operatorname{Beta}(1,1)$, from Proposition~\ref{prop:beta_prior}\label{table:basin-posterior}}
{
\setlength{\tabcolsep}{4pt}
\makebox[\textwidth][c]{%
\begin{tabular*}{\textwidth}{@{\extracolsep{\fill}}lccccc}
\toprule
 & $\varepsilon=10^{-4}$ & $\varepsilon=10^{-3}$ & $\varepsilon=10^{-2}$ & $\varepsilon=5\times 10^{-2}$ & $\varepsilon=10^{-1}$ \\
\midrule
$n=10$   & $0.0011$ & $0.0109$ & $0.1047$ & $0.4312$ & $0.6862$ \\
$n=100$   & $0.0100$ & $0.0961$ & $0.6376$ & $0.9944$ & $1-2.39\times 10^{-5}$ \\
$n=1000$   & $0.0953$ & $0.6327$ & $1-4.27\times 10^{-5}$ & $1-5.03\times 10^{-23}$ & $1-1.57\times 10^{-46}$ \\
$n=10000$   & $0.6322$ & $1-4.51\times 10^{-5}$ & $1-2.23\times 10^{-44}$ & $1-1.64\times 10^{-223}$ & $1-2.40\times 10^{-458}$ \\
\bottomrule
\end{tabular*}
}
}
{}
\end{table}

\begin{table}[t]
\TABLE
{\small Posterior probability $\Pi(s_{\vx^*}=1\mid H_n)$ under spike-and-slab prior with $W=\operatorname{Beta}(1,1)$ and varying $p$, from \eqref{eq:z=1_posterior}\label{table:spike-slab}}
{
\setlength{\tabcolsep}{4pt}
\makebox[\textwidth][c]{%
\begin{tabular*}{\textwidth}{@{\extracolsep{\fill}}lcccc}
\toprule
 & $p=10^{-3}$ & $p=10^{-2}$ & $p=10^{-1}$ & $p=0.5$ \\
\midrule
$n=10$   & $0.0109$ & $0.100$ & $0.550$ & $0.917$ \\
$n=100$   & $0.0918$ & $0.505$ & $0.918$ & $0.990$ \\
$n=1000$   & $0.501$ & $0.910$ & $0.991$ & $0.999$ \\
$n=10000$   & $0.909$ & $0.990$ & $0.999$ & $1.000$ \\
\bottomrule
\end{tabular*}
}
}
{}
\end{table}

\begin{table}[t]
\TABLE
{\small Lower bounds for $\PP(K=1\mid H_n)$ under geometric and zero-truncated Poisson MFM priors over $K$, from Theorem~\ref{thm:K1_tight_rate_fixed}\label{table:MFM-bound}}
{
\footnotesize
\setlength{\tabcolsep}{3pt}
\makebox[\textwidth][c]{%
\begin{tabular*}{\textwidth}{@{\extracolsep{\fill}}lcccccccc}
\toprule
 & \multicolumn{4}{c}{Heavier tail: $\PP(K=k)=2^{-k}$} & \multicolumn{4}{c}{Lighter tail: $\PP(K=k)=\frac{e^{-1}}{(1-e^{-1})k!}$} \\
\cmidrule(lr){2-5}\cmidrule(lr){6-9}
$n$ & $\alpha=0.1$ & $\alpha=0.5$ & $\alpha=1$ & $\alpha=5$ & $\alpha=0.1$ & $\alpha=0.5$ & $\alpha=1$ & $\alpha=5$ \\
\midrule
$10$ & $0.552$ & $0.693$ & $0.787$ & $0.949$ & $0.679$ & $0.779$ & $0.847$ & $0.964$ \\
$100$ & $0.601$ & $0.825$ & $0.930$ & $1-3.90\times 10^{-4}$ & $0.713$ & $0.874$ & $0.949$ & $1-2.80\times 10^{-4}$ \\
$1000$ & $0.644$ & $0.902$ & $0.978$ & $1-1.38\times 10^{-6}$ & $0.745$ & $0.929$ & $0.984$ & $1-9.89\times 10^{-7}$ \\
$10000$ & $0.683$ & $0.945$ & $0.993$ & $1-4.40\times 10^{-9}$ & $0.772$ & $0.960$ & $0.995$ & $1-3.16\times 10^{-9}$ \\
\bottomrule
\end{tabular*}
}
}
{}
\end{table}

\section{Limits of Restart-Based Inference}\label{sec:limits}

While Sections~4--6 provide conditions under which restart-based inference is statistically coherent, the framework also clarifies intrinsic limitations of what restart evidence can reveal. This section highlights several such limitations.

\subsection{Interpreting the evidence}
As discussed in the Introduction, restart-based evidence is leveraged to derive a variety of conclusions. Some are positive conclusions, e.g.\ reproducibility or stability of the observed terminal outcome. These can be stated in terms of the size of the observed basin of attraction. As we saw in Section~\ref{sec:inference}, restart-based evidence provides strong evidence for such claims. On the other hand, some conclusions also have a negative component, e.g.\ that there are no other terminal outcomes, or that the terminal outcome is unique. Such claims are trickier to support with restart-based evidence.

The reason is the visibility limitation formalized by $\mu$-observable RTOs in Definition~\ref{def:RTO}. If an RTO has basin size zero under $\mu$, then it will almost surely never be observed under repeated random restarts, regardless of the number of runs performed. Consequently, repeated convergence to a single observed outcome cannot statistically rule out additional non-$\mu$-observable RTO. This is a limitation of the solver-restart experiment, not a sampling error that disappears with more data.

Additionally, while basins of attraction themselves are properties of the solver, their size depends on the law according to which initial conditions are randomized. When $X$ is compact, there is a natural choice of $\mu$ as the uniform distribution over $X$, but if $X$ is unbounded, there is no such natural choice. So basin sizes should be interpreted while keeping in mind the particular restart distribution $\mu$. Thus, even a complete solver supports problem-level claims only for solutions visible under the restart distribution.

\subsection{Numerical Inaccuracies and Observed Partitions}
Section~\ref{subsec:numerical} anchors numerical inaccuracies through the observed terminal map $g_{\vQ}^{\varepsilon}=\mathcal O_\varepsilon\circ g_{\vQ}$. The observation map can combine distinct idealized terminal outcomes into the same observed outcome. When this happens, restart evidence supports claims about the observed partition, not necessarily the idealized one.

The simplest example is when the idealized solver admits a continuum of terminal outcomes. Any realistic numerical approximation will coarsen this continuum into finitely many observed terminal outcomes. Hence the partition of initial conditions into basins of attraction is very different under $g_{\vQ}^{\varepsilon}$ compared to $g_{\vQ}$.

We are eschewing the question of whether two outcomes should be treated as ``the same'' or ``different'' beyond checking equality under the observation map. In practice, we treat extremely close solutions as the same, while completely different solutions are treated as separate. Practitioners have to draw the line of distinction somewhere, and we do not provide a guideline for how this must be done.
 
\subsection{Dependence Between Restarts}

Section~\ref{sec:applications} formalizes a particular experiment: independent randomized restarts, under which observed terminal outcomes behave approximately as i.i.d.\ samples from the induced distribution over RTO.

There are also ``restart-like'' data, both designed by practitioners (e.g.\ warm starts, adaptive perturbations, random walks over initial conditions) and generated by nature (e.g.\ repeated market or traffic outcomes under evolving fundamentals). Such data fit our framework only to the extent that they satisfy the effective independence assumption. When initial conditions are correlated, the nominal number of observations can substantially overstate the effective sample size, so repeated observations of the same outcome need not provide independent evidence about multiplicity.

\section{Conclusion}
\label{sec:conclusion}
Random restarts are often used to support several claims at once: robustness to initialization, dominance under the solver, and sometimes uniqueness. This paper separates these interpretations. Repeated convergence directly informs basin-size claims about the observed reachable terminal outcome (RTO); it is weaker evidence for uniqueness because uniqueness also requires ruling out other RTOs with small or zero restart probability.

We identify structural assumptions under which restart evidence reduces solver geometry to categorical data over terminal outcomes. Under this reduction, posterior beliefs about the observed basin size concentrate exponentially fast, whereas posterior beliefs about uniqueness concentrate only polynomially under the priors studied here.

The application demonstrates this reduction for the Morrow--Skerlos fixed-point solver by connecting solver geometry to categorical data over terminal outcomes.

As emphasized throughout, the conclusions remain solver-level unless additional structure links the solver's RTOs to the problem's solutions; completeness, $\mu$-observability, numerical resolution, and dependence across restarts all affect that upgrade. We have focused on the event that all restarts produce the same outcome. Conditioning on richer events, such as observing multiple RTOs, requires nontrivial choices about priors on latent outcome structure and is left for future work.

\bibliographystyle{plainnat}
\bibliography{references}

\appendix
\section{Dynamic-Solver Arguments}

This appendix supplies the dynamic-solver details used in
Section~\ref{sec:model}. It is organized around the three properties needed for
the categorical reduction. First, locally well-posed projected trajectories give
a terminal map. Second, that terminal map is measurable. Third, ordinary
terminal outcomes lie among the roots of the projected vector field.
The first subsection proves the sufficient condition used in
Lemma~\ref{lemma:Lipschitz}: locally Lipschitz dynamics generate a unique
maximal local projected flow. The second subsection proves that the terminal
map $g_{\vQ}$ is measurable even though some trajectories may fail to exist for
all time or may fail to converge. The third subsection proves the root
containment behind Proposition~\ref{prop:finite-rto-roots}. This last argument
is included because $\QTC$ need not be continuous at the boundary, so standard
continuous-vector-field arguments such as Barbalat's lemma do not apply
directly.

\subsection{Well-Posedness Under Local Lipschitz Dynamics}\label{proof:lipschitz-local-flow}

This subsection proves Lemma~\ref{lemma:Lipschitz}. The global Lipschitz case
is a standard projected-dynamical-system theorem. The locally Lipschitz case
is obtained by the global through the following trick: for any compact
$K\subseteq X$ there exists a globally Lipschitz map $\widehat{\vQ}_K$ that
agrees with $\vQ$ on $K$. As the local behavior of the flow depends only on the
local dynamics, we can use the global flow $\widehat{\vG}_K$ of
$\widehat{\vQ}_K$ to define the local flow of $\vQ$ until the exit time from
$K$. By patching together these local flows, we can construct a maximal local
flow for $\vQ$ on $X$. The uniqueness and continuity properties follow from the
uniqueness and continuity of the global flows on each compact region and the
consistency of these flows before exit, verifying
Assumption~\ref{ass:well-behaved-proj}.

\begin{lemma}[Global well-posedness under global Lipschitzness]\label{lemma:global-lip-flow}
If $\vQ$ is globally Lipschitz on $X$, then the projected dynamics generated by
$\vQ$ admit a unique global flow
\[
\vG:X\times[0,\infty)\to X
\]
that is continuous in $(\vx,t)$.
\end{lemma}
\begin{proof}
Global Lipschitzness implies the one-sided Lipschitz and linear-growth
conditions required by Theorem~2.5 of \citet{nagurney2012projected} on a closed
convex set. Their projected system uses the opposite sign convention, so after
replacing their vector field by ours the theorem applies to
$\dot \vx=\Pi_{\TC_X(\vx)}\vQ(\vx)$. It gives a unique global solution and
continuous dependence on initial conditions, uniformly on compact time
intervals. Together with continuity of each trajectory in time, this yields
continuity of the flow map $(\vx,t)\mapsto \vG(\vx,t)$. In the notation of
Assumption~\ref{ass:well-behaved-proj}, the global case corresponds to
$\tau(\vx)=\infty$ for every $\vx\in X$.
\end{proof}

We now pass from global Lipschitzness to local Lipschitzness by localization.
The argument has five steps. First, Lemma~\ref{lemma:global-lip-flow} covers the
case in which $\vQ$ is globally Lipschitz. Second, for each bounded region
$K_R$, Lemma~\ref{lemma:localized-global-systems} constructs a globally
Lipschitz auxiliary dynamics $\widehat{\vQ}_{K_R}$ that agrees with $\vQ$ on
$K_R$; this auxiliary dynamics has a global flow
$\widehat{\vG}_{K_R}$. Third, Lemma~\ref{lemma:localized-consistency} shows that
these localized flows are compatible before exit, so the construction does not
depend on arbitrary choices of extensions. Fourth,
Lemma~\ref{lemma:local-flow-construction} patches the localized flows into a
single local solution map on all of $X$. Fifth,
Lemma~\ref{lemma:maximal-local-flow} shows that the patched solution map is the
unique maximal local flow of the original projected dynamics.

\begin{lemma}[Lipschitz extensions on bounded regions]\label{lemma:localized-global-systems}
Suppose $\vQ$ is locally Lipschitz. For each $R>0$, there exists a globally
Lipschitz map $\widehat{\vQ}_{K_R}:X\to\R^d$ such that
$\widehat{\vQ}_{K_R}=\vQ$ on $K_R:=X\cap\overline B_R$($B_R$ is the open ball of radius $R$ centered at the origin). The projected dynamics
generated by $\widehat{\vQ}_{K_R}$ have a unique global flow
$\widehat{\vG}_{K_R}:X\times[0,\infty)\to X$ that is continuous in $(\vx,t)$.
\end{lemma}
\begin{proof}
Local Lipschitzness is understood relative to $X$; by compactness, a finite
subcover of local Lipschitz neighborhoods gives a single Lipschitz constant for
$\vQ|_{K_R}$. Applying the McShane extension theorem coordinate by coordinate
gives a globally Lipschitz map
\[
\widehat{\vQ}_{K_R}:X\to\R^d
\]
possibly with a larger Lipschitz constant, such that
$\widehat{\vQ}_{K_R}=\vQ$ on $K_R$. The conclusion follows from
Lemma~\ref{lemma:global-lip-flow} applied to $\widehat{\vQ}_{K_R}$.
\end{proof}

For the rest of this subsection, fix the maps $\widehat{\vQ}_{K_R}$ and flows
$\widehat{\vG}_{K_R}$ from Lemma~\ref{lemma:localized-global-systems}. Define
the exit time
\[
\sigma_{K_R}(\vx):=\inf\{t\ge 0:\|\widehat{\vG}_{K_R}(\vx,t)\|\ge R\},
\]
with the convention $\inf\varnothing=\infty$.
Thus $\sigma_{K_R}(\vx)$ is the first time the global flow generated by
$\widehat{\vQ}_{K_R}$ leaves the ball $B_R$ used to define the
localization. Before that time, the localized vector field agrees with the
original vector field.

The localized extensions are chosen separately, so their global flows need not
agree after they leave the region where they both coincide with $\vQ$. The next
lemma records the compatibility needed for patching: before a smaller
localization exits, its flow is exactly the restriction of any larger
localization.

\begin{lemma}[Agreement before exit]\label{lemma:localized-consistency}
Fix $R>0$, $\vx\in X$, and $T<\sigma_{K_R}(\vx)$. Then
$\widehat{\vG}_{K_R}(\vx,\cdot)$ solves the original projected dynamics on
$[0,T]$.

Consequently, for every $S>R$,
\[
\widehat{\vG}_{K_R}(\vx,t)=\widehat{\vG}_{K_S}(\vx,t)
\qquad \text{for all }t\in[0,T].
\]
\end{lemma}
\begin{proof}
Since $T<\sigma_{K_R}(\vx)$,
$\widehat{\vG}_{K_R}(\vx,t)\in K_R$ for every $t\in[0,T]$. Write
$\vy(t):=\widehat{\vG}_{K_R}(\vx,t)$. Hence
\[
\widehat{\vQ}_{K_R}(\vy(t))=\vQ(\vy(t))
\qquad \text{for all }t\in[0,T].
\]
Therefore $\widehat{\vG}_{K_R}(\vx,\cdot)$ solves the original projected
dynamics on $[0,T]$.

Now let $S>R$. Since $K_R\subseteq K_S$, the same trajectory remains in the
region $K_R$, where both localized vector fields agree with $\vQ$. Therefore
along $\vy(\cdot)$, both
$\widehat{\vQ}_{K_R}(\vy(t))=\vQ(\vy(t))$ and
$\widehat{\vQ}_{K_S}(\vy(t))=\vQ(\vy(t))$ hold for all $t\in[0,T]$.
Thus the path $\vy(\cdot)$ also solves the globally Lipschitz projected
dynamics generated by $\widehat{\vQ}_{K_S}$ on $[0,T]$, with initial condition
$\vy(0)=\vx$. The flow
$\widehat{\vG}_{K_S}(\vx,\cdot)$ is the unique solution of that same
$\widehat{\vQ}_{K_S}$-system from the same initial condition. Hence
\[
\vy(t)=\widehat{\vG}_{K_S}(\vx,t)
\qquad \text{for all }t\in[0,T].
\]
Since $\vy(t)=\widehat{\vG}_{K_R}(\vx,t)$, the desired agreement follows.
\end{proof}

\begin{lemma}[Constructing the maximal local flow]\label{lemma:local-flow-construction}
Define
\[
\tau(\vx):=\sup_{R>\|\vx\|}\sigma_{K_R}(\vx).
\]
If $t<\tau(\vx)$, then some $R>\|\vx\|$ satisfies
$t<\sigma_{K_R}(\vx)$; otherwise all exit times in the supremum would be at
most $t$, contradicting $t<\tau(\vx)$.
For $0\le t<\tau(\vx)$, choose $R>\|\vx\|$ such that
$t<\sigma_{K_R}(\vx)$ and set
\[
\vG(\vx,t):=\widehat{\vG}_{K_R}(\vx,t).
\]
Then $\tau(\vx)>0$, $\vG$ is well defined on
$D:=\{(\vx,t):0\le t<\tau(\vx)\}$, and $\vG(\vx,\cdot)$ solves the original
projected dynamics on every compact subinterval of $[0,\tau(\vx))$.
\end{lemma}
\begin{proof}
The lifetime is strictly positive: if $R>\|\vx\|$, continuity of
$t\mapsto \widehat{\vG}_{K_R}(\vx,t)$ at $t=0$ gives
$\sigma_{K_R}(\vx)>0$.

Lemma~\ref{lemma:localized-consistency} shows that the value of
$\vG(\vx,t)$ is independent of which such localization is chosen. To see that
$\vG(\vx,\cdot)$ solves the original dynamics on compact subintervals, fix
$T<\tau(\vx)$. By the same supremum argument, choose $R>\|\vx\|$ with
$T<\sigma_{K_R}(\vx)$. Then
$\vG(\vx,t)=\widehat{\vG}_{K_R}(\vx,t)$ for all $t\in[0,T]$ by
Lemma~\ref{lemma:localized-consistency}, and
$\widehat{\vG}_{K_R}(\vx,\cdot)$ solves the original projected dynamics on
$[0,T]$.
\end{proof}

\begin{lemma}[Maximality and uniqueness]\label{lemma:maximal-local-flow}
The lifetime $\tau$ from Lemma~\ref{lemma:local-flow-construction} is the
maximal existence time. Moreover, any local solution of the original projected
dynamics agrees with $\vG$ on its interval of existence; hence $\vG$ is the
unique maximal local solution map.
\end{lemma}
\begin{proof}
Any other local solution stays in some bounded region on each compact time
interval. On that region it is also a solution of one globally Lipschitz
localized system, where uniqueness is already known.
Let $\vy:[0,T]\to X$ be any solution of the original
projected dynamics with $\vy(0)=\vx$. Since $\vy([0,T])$ is compact, choose
$R$ such that
\[
\sup_{0\le t\le T}\|\vy(t)\|<R.
\]
Then $\widehat{\vQ}_{K_R}=\vQ$ along $\vy([0,T])$, so $\vy$ also solves the
globally Lipschitz projected dynamics generated by $\widehat{\vQ}_{K_R}$. By
uniqueness for the $\widehat{\vQ}_{K_R}$-system,
\[
\vy(t)=\widehat{\vG}_{K_R}(\vx,t)\qquad\text{for all }t\in[0,T].
\]
Because this path agrees with $\vy$ and stays strictly inside $B_R$ on $[0,T]$,
we have $T<\sigma_{K_R}(\vx)$, and therefore $T<\tau(\vx)$. This proves
uniqueness on each compact interval on which a solution exists. Applying the
same argument on each compact interval $[0,T]\subset[0,S)$ shows that every
solution on $[0,S)$ satisfies $S\le \tau(\vx)$. Since
Lemma~\ref{lemma:local-flow-construction} already constructs a solution on
$[0,\tau(\vx))$, $\tau(\vx)$ is the maximal existence time.
\end{proof}

\begin{lemma}[Regularity of $\tau$ and $\vG$]\label{lemma:lsc-flow-continuity}
The maximal existence time $\tau$ from Lemma~\ref{lemma:local-flow-construction} is
lower semicontinuous, and the local flow $\vG$ is continuous on
$D:=\{(\vx,t):0\le t<\tau(\vx)\}$.
\end{lemma}
\begin{proof}
We first prove lower semicontinuity of $\tau$. It suffices to show that, for
each $T\ge 0$, the survival set
\[
A_T:=\{\vx\in X:\tau(\vx)>T\}
\]
is open in $X$. Fix $T\ge 0$ and $\vx_0\in A_T$. Then
$T<\tau(\vx_0)$, so by the definition of $\tau$ there exists
$R>\|\vx_0\|$ such that $T<\sigma_{K_R}(\vx_0)$. Then
\[
\sup_{0\le t\le T}\|\widehat{\vG}_{K_R}(\vx_0,t)\|<R.
\]
By continuity of the global flow $\widehat{\vG}_{K_R}$, uniformly over the
compact interval $[0,T]$, the same strict inequality holds for all initial
conditions $\vx$ in a neighborhood of $\vx_0$. Hence
$T<\sigma_{K_R}(\vx)\le \tau(\vx)$ for all such $\vx$, so this neighborhood is
contained in $A_T$. Therefore $A_T$ is open, proving that $\tau$ is lower
semicontinuous.

For continuity of $\vG$, fix $(\vx_0,t_0)\in D$. Choose
$T\in(t_0,\tau(\vx_0))$ and then choose $R>\|\vx_0\|$ with
$T<\sigma_{K_R}(\vx_0)$. By the preceding argument, for all $\vx$ near $\vx_0$
and all $t$ near $t_0$, the local flow agrees with the continuous global flow
$\widehat{\vG}_{K_R}$. Therefore $\vG$ is continuous on $D$.
\end{proof}

\lemLipschitz*
\begin{proof}[Proof of Lemma~\ref{lemma:Lipschitz}]
The global Lipschitz statement is Lemma~\ref{lemma:global-lip-flow}. For locally
Lipschitz $\vQ$, Lemmas~\ref{lemma:localized-global-systems}
and~\ref{lemma:localized-consistency} construct compatible global
localizations. Lemma~\ref{lemma:local-flow-construction} patches them into a
local flow. Lemma~\ref{lemma:maximal-local-flow} gives maximality and
uniqueness, and Lemma~\ref{lemma:lsc-flow-continuity} gives the lower
semicontinuous lifetime and continuity on the local-flow domain. These are
exactly the requirements in Assumption~\ref{ass:well-behaved-proj}.
\end{proof}

\subsection{Measurability of Dynamic Solvers}

\propMeasurable*
\begin{proof}\label{proof:measurable}
We prove measurability by decomposing $X$ according to the three possible
terminal behaviors of the local flow. Let $C_1$ be the set of initial
conditions whose trajectories have finite lifetime, let $C_2$ be the set of
initial conditions whose trajectories survive forever but do not converge, and
let $C_3$ be the set of initial conditions whose trajectories survive forever
and converge. On $C_1\cup C_2$ the dynamic solver returns $\dagger$; on $C_3$
it returns the limit of the trajectory. Thus, once we show that
$C_1,C_2,C_3$ are Borel and that the limit map on $C_3$ is Borel,
measurability follows from the representation
$g_{\vQ}(\vx)=\dagger$ on $C_1\cup C_2$ and
$g_{\vQ}(\vx)=L(\vx):=\lim_{t\to\infty}\vG(\vx,t)$ on $C_3$.

We now define these sets formally. The infinite-survival set is
\[
A_\infty:=\{\vx:\tau(\vx)=\infty\}=\bigcap_{n=1}^\infty A_n
\]
where $A_n:=\{\vx\in X:\tau(\vx)>n\}$. By lower semicontinuity of $\tau$, each
$A_n$ is open in $X$, so $A_\infty$ is Borel.

On $A_\infty$, the flow $\vG(\vx,t)$ is defined for every $t\ge 0$. This
restriction is important: outside $A_\infty$, the values $\vG(\vx,q)$ are not
defined for all large rational $q$, and replacing missing values by $\dagger$
would not give coordinate functions in $\R^d$. For $i\in\{1,\dots,d\}$, write
$G_i(\vx,t)$ for the $i$-th coordinate of $\vG(\vx,t)$ and define
extended-real maps on $A_\infty$ by
\begin{align*}
\ell_i^+(\vx)
&:= \inf_{T\in\N}\ \sup_{\substack{q\in\Q\\ q\ge T}} G_i(\vx,q),\\
\ell_i^-(\vx)
&:= \sup_{T\in\N}\ \inf_{\substack{q\in\Q\\ q\ge T}} G_i(\vx,q).
\end{align*}
For each rational $q\ge 0$, the map $\vx\mapsto G_i(\vx,q)$ is continuous on
$A_\infty$ because $A_\infty\subseteq A_q$ and $\vG(\cdot,q)$ is continuous on
$A_q$. Hence countable suprema and infima make $\ell_i^+$ and $\ell_i^-$
measurable on $A_\infty$. These maps are exactly
$\limsup_{t\to\infty}G_i(\vx,t)$ and
$\liminf_{t\to\infty}G_i(\vx,t)$, because
$t\mapsto G_i(\vx,t)$ is continuous on $[0,\infty)$ and suprema and infima over
real times can be taken over rational times.

The convergence set is
\[
\Xi:=A_\infty\cap
\bigcap_{i=1}^d
\{\vx:\ell_i^+(\vx)=\ell_i^-(\vx)\in\R\}.
\]
It is Borel in $A_\infty$, and therefore in $X$, because the defining
coordinate conditions are Borel. On $\Xi$, the coordinate limits exist. Define
their $i$-th coordinate by
$L_i(\vx):=\lim_{t\to\infty}G_i(\vx,t)=\ell_i^+(\vx)=\ell_i^-(\vx)$, and set
$L(\vx):=(L_1(\vx),\dots,L_d(\vx))$.
Each $L_i$ is measurable on $\Xi$, so $L:\Xi\to\R^d$ is Borel. Since $X$ is
closed and $\vG(\vx,t)\in X$, we have $L(\vx)\in X$ for every $\vx\in\Xi$.

Define the condition sets
\[
C_1:=X\setminus A_\infty,\qquad
C_2:=A_\infty\setminus\Xi,\qquad
C_3:=\Xi.
\]
By the definition of the dynamic solver, $g_{\vQ}|_{C_3}=L$.

For every Borel $B\subseteq X$, $g_{\vQ}^{-1}(B)=C_3\cap L^{-1}(B)$ is Borel
in $X$, while $g_{\vQ}^{-1}(\{\dagger\})=C_1\cup C_2=X\setminus C_3$ is Borel.
Hence
$g_{\vQ}:X\to X\sqcup\{\dagger\}$ is measurable.
\end{proof}

\subsection{Terminal Outcomes and Roots}

We now prove the containment underlying Proposition~\ref{prop:finite-rto-roots}.
Throughout this subsection, $\mR(\QTC)$ denotes the zero set of the projected
vector field $\QTC:X\to\R^d$.

The claim is intuitive: if a trajectory converges to an ordinary terminal
outcome, the vector field cannot keep pushing it in any feasible direction at the
limit. The proof uses the normal-cone formulation of projected dynamics and a
time-average argument. This avoids assuming that $\QTC$ is continuous at
boundary points.

\begin{lemma}\label{lemma:containment}
A dynamic solver $g_{\vQ}$ induced by $\vQ$ satisfies the following containment:
\[
g_{\vQ}(X) \subseteq \mR(\QTC) \cup \{\dagger\}.
\]
\end{lemma}
\begin{proof}\label{proof:containment}
Fix an initial condition $\vx_0\in X$ and write
$\vx(t):=\vG(\vx_0,t)$ for the corresponding projected trajectory. If the
system does not converge, then $g_{\vQ}(\vx_0)=\dagger$ and the claim is
immediate. Hence assume $g_{\vQ}(\vx_0)=\vx^*\in X$, so $\vx(t)\to \vx^*$.

\emph{Step 1: use the normal-cone inequality along the trajectory.}
For a.e. $t$, using the normal cone formulation, the projected dynamics satisfies:
\[
\vQ(\vx(t)) - \dot \vx(t)\in N_X(\vx(t))
\]
Throughout the proof, we consider times $t$ at which this relation holds. By definition of normal cone, for every $\vy\in X$,
\[
\begin{aligned}
&\langle \vQ(\vx(t))-\dot \vx(t),\,\vy-\vx(t)\rangle \le 0,\\
&\Rightarrow \langle \vQ(\vx(t)),\,\vy-\vx(t)\rangle\\
&\qquad\le \langle \dot \vx(t),\,\vy-\vx(t)\rangle.
\end{aligned}
\]
Notice
\[
\frac{d}{dt} \left(-\frac12 \|\vy-\vx(t)\|^2 \right)
=\langle \dot \vx(t),\vy-\vx(t)\rangle .
\]
Substituting this identity gives
\[
\langle \vQ(\vx(t)),\vy-\vx(t)\rangle
\le -\frac12\frac{d}{dt}\|\vy-\vx(t)\|^2
\]

\emph{Step 2: average over time.}
Integrating from $0$ to $T$ and dividing by $T$ yields
\[
\begin{aligned}
&\frac1T\int_0^T \langle \vQ(\vx(t)),\vy-\vx(t)\rangle\,dt\\
&\qquad\le \frac{\|\vy-\vx(0)\|^2-\|\vy-\vx(T)\|^2}{2T}.
\end{aligned}
\]
Taking $\limsup_{T\to\infty}$ on both sides, we get
\[
\begin{aligned}
&\limsup_{T\to\infty}
\frac1T\int_0^T \langle \vQ(\vx(t)),\vy-\vx(t)\rangle\,dt\\
&\qquad\le \limsup_{T\to\infty}
\frac{\|\vy-\vx(0)\|^2-\|\vy-\vx(T)\|^2}{2T}\\
&\qquad=0.
\end{aligned}
\]

Since $\vx(T)\to\vx^*$, the right-hand side is bounded divided by $T$ and hence
has limit $0$.

\emph{Step 3: pass to the limit.}
As \(\vQ\) and \(\vx\) are continuous, the scalar function
\(f(t):=\langle \vQ(\vx(t)),\vy-\vx(t)\rangle\) converges to
\(f(\infty):=\langle \vQ(\vx^*),\vy-\vx^*\rangle\) as \(t\to\infty\).
Moreover, since \(\vx(t)\to\vx^*\), the trajectory \(\vx(\cdot)\) is bounded,
and continuity of \(\vQ\) on the trajectory implies that \(f\) is bounded on
$[0,\infty)$. Hence the rescaled functions $s\mapsto f(Ts)$ on $[0,1]$ are
dominated by an integrable constant and converge pointwise a.e. to $f(\infty)$.
By the dominated convergence theorem (equivalently, the standard Ces\`aro averaging
lemma),
\begin{align*}
\lim_{T\to\infty}\frac{1}{T}\int_0^T f(t)\,dt
    &=\lim_{T\to\infty}\int_0^1 f(Ts)\,ds \\
    &=\int_0^1 f(\infty)\,ds \\
    &=f(\infty).
\end{align*}

Thus the time average converges to \(\langle \vQ(\vx^*),\vy-\vx^*\rangle\) as claimed.

Combining this with the previous inequality, we get
\begin{align*}
\langle \vQ(\vx^*),\vy-\vx^*\rangle \leq 0
\end{align*}
As this holds for every $\vy\in X$, we have
$\vQ(\vx^*) \in N_X(\vx^*)$. By definition of projected dynamics, this is
equivalent to $\QTC(\vx^*)=0$. Hence $\vx^* \in \mR(\QTC)$, proving the
containment.
\end{proof}

\section{Prevalence Arguments}
\label{appendix:prevalence}

This appendix proves Proposition~\ref{prop:e-submersion-finite-attractors}. The
proof has three steps. First, we split roots of $\QTC$ into continuous interior
and boundary root problems. Second, we use parametric transversality to show
that these root sets are discrete for prevalent dynamics. Third, we combine
discreteness with compact containment from Assumption~\ref{ass:kkq} to obtain
finiteness.

\begin{remark}
    With $C^2$ boundary, the normal vector $\vn:\partial X \to \R^d$ is well-defined and $C^1$. So we can write $\QTC$ as:
    \[
        \QTC(\vx)=
        \begin{cases}
            \vQ(\vx) & \vx \in \inter X \\
            \QTCP(\vx) & \vx \in \partial X
        \end{cases}
    \]
    where, for $\vx\in\partial X$,
    \[
        \QTCP(\vx)=
        \vQ(\vx) - \langle \vQ(\vx), \vn(\vx) \rangle_+ \vn(\vx),
    \]
    and the $+$ subscript denotes the positive part (maximum of the argument and zero).
\end{remark}

The final step of the argument uses the following compactness fact.

\begin{lemma}[Closed discrete sets in compact sets are finite]\label{lemma:closed-discrete-finite}
Let $Y$ be a topological space, and let $S\subseteq Y$ be closed in $Y$ and
discrete in $Y$. If there is a compact set $K\subseteq Y$ such that
$S\subseteq K$, then $S$ is finite.
\end{lemma}
\begin{proof}
Since $S$ is closed in $Y$, $S\cap K=S$ is closed in $K$. Hence $S$ is compact
in the subspace topology. Since $S$ is discrete, for each $s\in S$ there is an
open set $U_s\subseteq Y$ such that $U_s\cap S=\{s\}$. The collection
$\{U_s\cap S:s\in S\}$ is an open cover of the compact space $S$ by singletons,
so it has a finite subcover. Hence $S$ is finite.
\end{proof}

In Section~\ref{sec:model}, we showed that $g_{\vQ}(X)$ is finite iff $\QTC$
has finitely many roots. The useful fact about roots of a continuous function is
that they form a closed set. Although $\QTC$ is not continuous, its roots are
contained in the union of the roots of two continuous maps.

\begin{lemma}\label{lemma:break-into-two}
	   \[\mR(\QTC) \subseteq \mR(\vQ) \cup \mR(\QTCP)\]
	   where $\mR(\vQ)$ is closed in $X$ and $\mR(\QTCP)$ is closed in $\partial X$.
\end{lemma}
\begin{proof}
Take any $\vx\in \mR(\QTC)$, i.e. $\QTC(\vx)=0$. If $\vx\in \inter X$, then $\QTC(\vx)=\vQ(\vx)$, so $\vQ(\vx)=0$ and hence $\vx\in \mR(\vQ)$. If instead $\vx\in \partial X$, then $\vx\in \mR(\QTCP)$ by definition. Therefore $\mR(\QTC) \subseteq \mR(\vQ) \cup \mR(\QTCP)$.

As $\vQ$ is continuous, $\mR(\vQ)$ is closed. Under Assumption~\ref{ass:kkq}, $X$ has $C^2$ boundary. For $\vx\in\partial X$,
\[\QTCP(\vx) = \vQ(\vx) - \langle \vQ(\vx), \vn(\vx) \rangle_+ \vn(\vx)\]
As the boundary is $C^2$, $\vn$ is $C^1$, making $\QTCP$ continuous on
$\partial X$ and $\mR(\QTCP)$ closed in $\partial X$.
\end{proof}

What remains is having a condition that implies discreteness of these root sets.
Once we have discreteness, Assumption~\ref{ass:kkq} provides compact containment
and Lemma~\ref{lemma:closed-discrete-finite} gives finiteness. When a function
$f$ is $C^1$, a sufficient condition for roots of $f$ to be isolated is
transversality of $f$ to $0$: $f \pitchfork 0$. As transversality works well
with prevalence, we phrase discreteness in terms of transversality to $0$. This
works directly for $\vQ$, which is $C^1$. The boundary map $\QTCP$, however, is
not $C^1$ because it involves $\langle \vQ(\vx), \vn(\vx) \rangle_+$.

To avoid the nonsmooth positive-part term in $\QTCP$, we encode boundary roots
using the following $C^1$ map:
\[
\begin{aligned}
    \tilde{\vQ}:\partial X \times [0,\infty) \to \R^{d} \\
    \tilde{\vQ}(\vx, \alpha) = \vQ(\vx) - \alpha \vn(\vx)
\end{aligned}
 \]
The following lemma shows why $\tilde{\vQ} \pitchfork 0$ implies discreteness of $\mR(\QTCP)$.

\begin{lemma}\label{lemma:discrete-preserves}
    If $\tilde{\vQ} \pitchfork 0$, then $\mR(\QTCP)$ is discrete in $\partial X$.
\end{lemma}
\begin{proof}
    If $\tilde{\vQ} \pitchfork 0$, then $\mR(\tilde{\vQ})$ is discrete in
    $\partial X\times[0,\infty)$. Moreover, $\vx\in\mR(\QTCP)$ iff
    \[
    (\vx,\alpha(\vx))\in\mR(\tilde{\vQ}),
    \qquad
    \alpha(\vx):=\langle \vQ(\vx),\vn(\vx)\rangle_+.
    \]

    Suppose for contradiction that $\mR(\QTCP)$ is not discrete in $\partial X$.
    Then there exists a sequence of distinct points $\vx_n\in \mR(\QTCP)$ and a limit point
    $\vx^*=\lim_{n\to\infty} \vx_n\in \mR(\QTCP)$.
    By continuity of $\alpha(\cdot)$,
    $(\vx_n,\alpha(\vx_n))\to(\vx^*,\alpha(\vx^*))$. Each pair lies in
    $\mR(\tilde{\vQ})$, and the pairs are distinct because the $\vx_n$ are
    distinct. Thus $\mR(\tilde{\vQ})$ has an accumulation point, contradicting
    discreteness.
    Therefore $\mR(\QTCP)$ is discrete.
\end{proof}

So to prove discreteness of $\mR(\vQ)$ and $\mR(\QTCP)$, it suffices to prove $\vQ \pitchfork 0$ and $\tilde{\vQ} \pitchfork 0$. The following lemma connects prevalence and transversality.

\begin{restatable}[Parametric transversality from a submersion]{lemma}{lemParametricSubmersion}\label{lem:parametric_submersion}
Let $Z$ be a finite-dimensional $C^{r}$ manifold, possibly with boundary, and let
$P$ be a finite-dimensional real vector space (identified with $\R^{m}$, with
Lebesgue measure $\lambda_{P}$). Let $\vF:P \times Z\to \R^{n}$ be a $C^{r}$ map
with $r>\max(\dim Z-n,0)$ such that $\vF$ is a submersion, i.e. for every $(\vp,z)\in P\times Z$,
\[
D\vF_{(\vp,z)}:T_{\vp}P\times T_{z}Z\to \R^{n}
\quad\text{is surjective.}
\]
For each $\vp \in P$ write $\vF_{\vp}(z):=\vF(\vp,z)$.
Then for $\lambda_{P}$-a.e.\ $\vp\in P$ we have $\vF_{\vp}\pitchfork 0$.
\end{restatable}

\begin{remark}\label{remark:parametric-regularity}
    In particular, if $\dim Z=n$ then
    $r>\max(\dim Z-n,0)$ reduces to $r>0$, so $C^{1}$ regularity suffices. This
    is the typical case for the applications in this paper; when $\dim Z>n$ the
    stronger $C^{r}$ assumption above is required to apply Sard's theorem to
    the projection used in the proof (see \cite{GuilleminPollack1974}).
\end{remark}

The idea is to apply Sard's theorem to the projection from the zero set
$M:=\vF^{-1}(0)$ onto the parameter space $P$. Regular values of this projection
are exactly the parameters for which the fiber map $\vF_{\vp}$ is transverse to
$0$.

\begin{proof}\label{proof:parametric_submersion}
Let $q:=\dim Z$. Since $\vF$ is a submersion, $0$ is a regular value of $\vF$.
Using the preimage theorem, $M:=\vF^{-1}(0)\subset P\times Z$ is a
$\dim(P)+q-n$-dimensional $C^{1}$ submanifold (possibly with boundary) with
$T_M = \operatorname{ker}(D\vF)$.

Let $\pi:P\times Z\to P$ be the projection $\pi(\vp,z)=\vp$ and set
$\pi^{M}:=\pi|_{M}$. Intuitively, values of $\pi^{M}$ are $\vp$ such that for
some $z$, $\vF(\vp,z)=0$. By Sard's theorem, the set of critical values of
$\pi^{M}$ has $\lambda_{P}$-measure zero. We show that regular values of
$\pi^{M}$ correspond exactly to $\vp$ such that $\vF_{\vp}\pitchfork 0$.

\[
\begin{aligned}
\ker D \pi^{M}_{(\vp,z)}
&=\{\vvec\in T_{(\vp,z)}M:D \pi^{M}_{(\vp,z)}(\vvec)=0\}\\
&=\{(\vvec_{\vp},\vvec_{z})\in T_{(\vp,z)}M:\vvec_{\vp}=0\}\\
&=\{(0,\vvec_{z})\in T_{(\vp,z)}M\}\\
&=\{(0,\vvec_{z})\in\ker(D\vF_{(\vp,z)})\}\\
&=\{(0,\vvec_{z}):D_{z}\vF_{(\vp,z)}(\vvec_{z})=0\}\\
&=\{(0,\vvec_{z}):\vvec_{z}\in\ker(D_{z}\vF_{(\vp,z)})\}.
\end{aligned}
\]

So $\dim(\ker D \pi^{M}_{(\vp,z)}) = \dim(\ker(D_{z} \vF_{(\vp,z)}))$.
Using the rank-nullity theorem on $D \pi^{M}_{(\vp,z)}$, we have:
\[
\begin{aligned}
&\rank(D \pi^{M}_{(\vp,z)})
+\dim(\ker D \pi^{M}_{(\vp,z)})\\
&\qquad=\dim(T_{(\vp,z)}M).
\end{aligned}
\]
\[
\begin{aligned}
&\Rightarrow \rank(D \pi^{M}_{(\vp,z)})
+\dim(\ker D \pi^{M}_{(\vp,z)})\\
&\qquad=\dim(P)+q-n.
\end{aligned}
\]
Doing the same on $D_{z} \vF$, we have:
\[
\begin{aligned}
\rank(D_{z} \vF_{(\vp,z)})+\dim(\ker D_{z} \vF_{(\vp,z)})
&=\dim(Z)\\
&=q.
\end{aligned}
\]
Subtracting the two equations cancels the kernels as they are equal giving us:
\[
\rank(D \pi^{M}_{(\vp,z)})-\rank(D_{z} \vF_{(\vp,z)})
=\dim(P)-n.
\]

Now if $\vp$ is a regular value of $\pi^{M}$, then
$\rank(D \pi^{M}_{(\vp,z)}) = \dim(P)$. Therefore
$\rank(D_{z} \vF_{(\vp,z)}) = n$, which is equivalent to
$\vF_{\vp} \pitchfork 0$.

On the other hand, assume $\vF_{\vp} \pitchfork 0$. Then
$\rank(D_{z} \vF_{(\vp,z)}) = n$. Using the previous equation, we have:
\[\rank(D \pi^{M}_{(\vp,z)}) = \dim(P)\]
So $D \pi^{M}_{(\vp,z)}$ is surjective, and hence $\vp$ is a regular value of $\pi^{M}$.
\end{proof}

\begin{remark}
    When $n=\dim Z$, $\vF_{\vp} \pitchfork 0$ means $0$ is a regular value of $\vF_{\vp}$. Using the preimage theorem, we have $\dim(\vF_{\vp}^{-1}(0)) = \dim Z - n = 0$. This implies $\vF_{\vp}^{-1}(0)$, i.e. the roots of $\vF_{\vp}$, are discrete in $Z$.
\end{remark}

To obtain transversality prevalently, we perturb the dynamics along a
finite-dimensional probe space.

\begin{definition}
    For a finite-dimensional subspace $P$ of $\mathcal{Q}$, for any $\vQ \in \mathcal{Q}$, define the following perturbed evaluation maps $e^{\vQ}:P\times X\to\R^{d}$ and $e^{\tilde{\vQ}}:P \times \partial X \times [0,\infty) \to \R^{d}$:
    \[
    \begin{aligned}[t]
        e^{\vQ}(\vp,\vx)&:=(\vQ+\vp)(\vx)\\
        e^{\tilde{\vQ}}(\vp,\vx,\alpha)&:=(\vQ+\vp)(\vx) - \alpha \vn(\vx)
    \end{aligned}
    \]
\end{definition}

We can use this to prove the following sufficient condition for prevalence of transversality.

\begin{theorem}[Prevalence via a submersion probe]\label{theorem:prevalence_submersion_probe}
Let $X$ have $C^2$ boundary. Let $\mathcal Q$ be a complete metric vector space of continuously differentiable functions from $X$ to $\R^d$ and let $P\subset \mathcal Q$ be a finite-dimensional subspace.

If for all $\vQ\in\mathcal Q$, the maps $e^{\vQ}$ and $e^{\tilde{\vQ}}$ are $C^{1}$ submersions, then for prevalent $\vQ \in \mathcal{Q}$, $\mR(\vQ)$ is discrete in $X$ while $\mR(\QTCP)$ is discrete in $\partial X$.
\end{theorem}

\begin{proof}
Fix any base dynamics $\vQ_0\in\mathcal Q$. Applying
Lemma~\ref{lem:parametric_submersion} to $e^{\vQ_0}$ gives a full-measure set of
perturbations $\vp\in P$ such that $(\vQ_0+\vp)(\cdot)\pitchfork 0$, and hence
$\mR(\vQ_0+\vp)$ is discrete in $X$. Applying the same lemma to
$e^{\tilde{\vQ}_0}$ gives a full-measure set of perturbations $\vp\in P$ such
that
\[
(\vx,\alpha)\mapsto(\vQ_0+\vp)(\vx)-\alpha\vn(\vx)
\]
is transverse to $0$. By Lemma~\ref{lemma:discrete-preserves},
$\mR((\vQ_0+\vp)^{\Pi_{\TC_X}}_{|\partial X})$ is discrete in $\partial X$ for
these perturbations. Intersecting the two full-measure sets gives full measure
in $P$. Since this holds for every base $\vQ_0\in\mathcal Q$, the set of
dynamics for which both root sets are discrete is prevalent.
\end{proof}

\begin{corollary}\label{cor:prevalent-finite-roots}
Assume $X$ satisfies Assumption~\ref{ass:convex-closed}, and let $\mathcal Q$
satisfy Assumption~\ref{ass:kkq}. If there exists a finite-dimensional subspace
$P\subseteq\mathcal Q$ such that, for every $\vQ\in\mathcal Q$, the perturbed
evaluation maps $e^{\vQ}$ and $e^{\tilde{\vQ}}$ are $C^1$ submersions, then for
prevalent $\vQ\in\mathcal Q$, $\QTC$ has finitely many roots.
\end{corollary}
\begin{proof}
By Theorem~\ref{theorem:prevalence_submersion_probe}, for prevalent
$\vQ\in\mathcal Q$, $\mR(\vQ)$ is discrete in $X$ and $\mR(\QTCP)$ is discrete
in $\partial X$. By Lemma~\ref{lemma:break-into-two},
\[
\mR(\QTC)\subseteq \mR(\vQ)\cup \mR(\QTCP).
\]
The two root sets on the right are closed in their respective domains by
Lemma~\ref{lemma:break-into-two}. Assumption~\ref{ass:kkq} places them inside
compact sets. Lemma~\ref{lemma:closed-discrete-finite} therefore implies that
both $\mR(\vQ)$ and $\mR(\QTCP)$ are finite. Hence $\mR(\QTC)$ is finite.
\end{proof}

\phantomsection\label{proof:e-submersion-finite-attractors}
\propESubmersionFiniteAttractors*
\begin{proof}[Proof of Proposition~\ref{prop:e-submersion-finite-attractors}]
The hypotheses of Proposition~\ref{prop:e-submersion-finite-attractors} are
exactly the hypotheses of Corollary~\ref{cor:prevalent-finite-roots}. Hence,
for prevalent $\vQ\in\mathcal Q$, the projected dynamics $\QTC$ has finitely
many roots.

By Lemma~\ref{lemma:containment},
\[
g_{\vQ}(X)\subseteq \mR(\QTC)\cup\{\dagger\}.
\]
Thus $g_{\vQ}(X)$ is contained in a finite set, so $|g_{\vQ}(X)|<\infty$.
Therefore, for prevalent $\vQ\in\mathcal Q$, the induced dynamic solver has
finitely many reachable terminal outcomes.
\end{proof}

\section{Bayesian Statistical Arguments}

\subsection{Basin-Size Concentration}
This subsection proves the auxiliary bounds used in Section~\ref{sec:inference}
to pass from repeated convergence to posterior concentration of the observed
basin size. Throughout this subsection,
$A_\varepsilon:=\{s_{\vx^*}\le 1-\varepsilon\}$.

\begin{lemma}[Bayes denominator bound]\label{lem:eta_bound}
Let $0 < \eta,\varepsilon < 1$ and $\Pi(A_\eta) < 1$. Then
\[
\Pi(A_\varepsilon \mid H_n(\vx^*))
\le
\frac{(1-\varepsilon)^{n}\,\Pi(A_\varepsilon)}
{(1-\eta)^{n}\,(1-\Pi(A_\eta))}.
\]
\end{lemma}

\begin{proof}
By definition of $A_\varepsilon$, on $A_\varepsilon$ we have
$s_{\vx^*}\le 1-\varepsilon$, so
\[
\PP(H_n(\vx^*)\cap A_\varepsilon)
\le (1-\varepsilon)^n\Pi(A_\varepsilon).
\]
Also, on $A_\eta^c$ we have $s_{\vx^*}>1-\eta$, hence
\[
\begin{aligned}
\PP(H_n(\vx^*))
&\ge \PP(H_n(\vx^*)\cap A_\eta^c)\\
&\ge (1-\eta)^n(1-\Pi(A_\eta)).
\end{aligned}
\]
Substituting these two bounds into
\[
\Pi(A_\varepsilon\mid H_n(\vx^*))
=\frac{\PP(H_n(\vx^*)\cap A_\varepsilon)}{\PP(H_n(\vx^*))}
\]
gives the result.
\end{proof}

\begin{lemma}[Support near one]\label{lem:supp_tail}
If $1 \in \operatorname{supp}(\Pi)$, then for every $\varepsilon\in(0,1)$
there exists $\eta\in(0,\varepsilon)$ such that $\Pi(A_\eta)<1$.
\end{lemma}

\begin{proof}
Choose any $\eta\in(0,\varepsilon)$. Since $1\in\operatorname{supp}(\Pi)$,
we have $\Pi((1-\eta,1])>0$. Since
$(1-\eta,1]\subseteq A_\eta^c$, it follows that
$\Pi(A_\eta^c)>0$, hence $\Pi(A_\eta)<1$.
\end{proof}

\phantomsection\label{proof:eta-bound}
\propEtaBound*
\begin{proof}[Proof of Proposition~\ref{prop:eta-bound}]
By Lemma~\ref{lem:supp_tail}, choose $\eta\in(0,\varepsilon)$ with
$\Pi(A_\eta)<1$. Lemma~\ref{lem:eta_bound} then gives
\[
\Pi(A_\varepsilon\mid H_n(\vx^*))
\le
\frac{(1-\varepsilon)^n\Pi(A_\varepsilon)}
{(1-\eta)^n(1-\Pi(A_\eta))}.
\]
Since $\eta<\varepsilon$, the ratio $(1-\varepsilon)/(1-\eta)$ is strictly
less than one, so the right-hand side decays exponentially in $n$.
\end{proof}

\propBasinPosteriorSize*
\begin{proof}\label{proof:basin-posterior-size}
We start from Lemma~\ref{lem:eta_bound}, which states that for any $\eta\in(0,1)$ with
$\Pi(A_\eta)<1$,
\begin{equation}\label{eq:eta_lemma}
\Pi(A_\varepsilon\mid H_n(\vx^*))
\;\le\;
\frac{(1-\varepsilon)^{n}\,\Pi(A_\varepsilon)}{(1-\eta)^{n}\,(1-\Pi(A_\eta))}.
\end{equation}

We lower bound the denominator. Since $A_\eta=\{s_{\vx^*}\le 1-\eta\}$, we have
\[
\begin{aligned}
1-\Pi(A_\eta)
&=\Pi(s_{\vx^*}>1-\eta)\\
&=\int_{1-\eta}^{1} f_s(s)\,ds.
\end{aligned}
\]
For $\eta\in(0,\delta)$, Assumption~\ref{ass:poly_density} yields
\[
\begin{aligned}
1-\Pi(A_\eta)
&\ge \int_{1-\eta}^{1} c(1-s)^{\kappa}\,ds\\
&=c\int_{0}^{\eta} u^{\kappa}\,du\\
&=\frac{c}{\kappa+1}\,\eta^{\kappa+1}.
\end{aligned}
\]
Plugging this into \eqref{eq:eta_lemma} gives, for $\eta\in(0,\delta)$,
\begin{equation}\label{eq:pre_choice_eta}
\Pi(A_\varepsilon\mid H_n(\vx^*))
\;\le\;
\frac{\kappa+1}{c}\,
\frac{(1-\varepsilon)^{n}\,\Pi(A_\varepsilon)}{(1-\eta)^{n}\,\eta^{\kappa+1}}
\end{equation}

Now choose $\eta=1/(n)$. For all $n$ such that $1/(n)<\delta$, we can apply
\eqref{eq:pre_choice_eta}. As $(1-1/n)^n$ is increasing, we can lower-bound it by its value for n=2.
\[
\begin{aligned}
\left(1-\frac{1}{n}\right)^{n}
&\ge \left(1-\frac{1}{2}\right)^{2}
=\frac{1}{4},\\
\text{equivalently}\qquad
\left(1-\frac{1}{n}\right)^{-n}
&\le 4.
\end{aligned}
\]
we obtain from \eqref{eq:pre_choice_eta}
\[
\begin{aligned}
\Pi(A_\varepsilon\mid H_n(\vx^*))
&\le \frac{\kappa+1}{c}\,
\left(1-\frac{1}{n}\right)^{-n}\\
&\qquad{}\cdot n^{\kappa+1}(1-\varepsilon)^{n}\Pi(A_\varepsilon)\\
&\le \frac{4(\kappa+1)}{c}\,
n^{\kappa+1}(1-\varepsilon)^{n}\Pi(A_\varepsilon).
\end{aligned}
\]
Finally, since $(1-\varepsilon)^{n} \le e^{-(n)\varepsilon}$ for $\varepsilon\in(0,1)$,
the second inequality in \eqref{eq:drop_eta_poly_bound} follows.
\end{proof}

\propBetaPrior*
\phantomsection\label{proof:beta_prior}
\begin{proof}[Proof of Proposition~\ref{prop:beta_prior}]
Under the $\operatorname{Beta}(\alpha,\beta)$ prior, the prior density of $s_{\vx^*}$ is
$p(s)\propto s^{\alpha-1}(1-s)^{\beta-1}$ on $(0,1)$.
The likelihood of $H_n$ given $s_{\vx^*}=s$ is $\PP(H_n\mid s)=s^n$.
Therefore the posterior density satisfies
$p(s\mid H_n)\propto s^n p(s)\propto s^{\alpha+n-1}(1-s)^{\beta-1}$, i.e.\
$s_{\vx^*}\mid H_n\sim\operatorname{Beta}(\alpha+n,\beta)$ with density
\[
\begin{aligned}
p(s\mid H_n)
&=\frac{1}{B(\alpha+n,\beta)}\,
s^{\alpha+n-1}(1-s)^{\beta-1},\\
&\hspace{8em}s\in(0,1).
\end{aligned}
\]
Hence
\[
\begin{aligned}
&\Pi\!\left(s_{\vx^*}\le 1-\varepsilon\mid H_n\right)\\
&\qquad=\int_{0}^{1-\varepsilon}
\frac{s^{\alpha+n-1}(1-s)^{\beta-1}}
{B(\alpha+n,\beta)}\,ds.
\end{aligned}
\]

Make the change of variables $u=1-s$. Then
\[
\begin{aligned}
&\Pi\!\left(s_{\vx^*}\le 1-\varepsilon\mid H_n\right)\\
&\qquad=\int_{\varepsilon}^{1}
\frac{u^{\beta-1}(1-u)^{\alpha+n-1}}
{B(\beta,\alpha+n)}\,du.
\end{aligned}
\]
For $u\in[\varepsilon,1]$ we have $(1-u)^{\alpha+n-1}\le (1-\varepsilon)^{\alpha+n-1}$, so
\[
\begin{aligned}
&\Pi\!\left(s_{\vx^*}\le 1-\varepsilon\mid H_n\right)\\
&\qquad\le \frac{(1-\varepsilon)^{\alpha+n-1}}
{B(\beta,\alpha+n)}
\int_{\varepsilon}^{1}u^{\beta-1}\,du.
\end{aligned}
\]
Evaluating the integral yields
\[
\begin{aligned}
&\Pi\!\left(s_{\vx^*}\le 1-\varepsilon\mid H_n\right)\\
&\qquad\le (1-\varepsilon)^{\alpha+n-1}\,
\frac{1-\varepsilon^{\beta}}{\beta\,B(\beta,\alpha+n)}.
\end{aligned}
\]

Using the identity
\[
\frac{1}{\beta\,B(\beta,\alpha+n)}
=\frac{\Gamma(\alpha+n+\beta)}{\Gamma(\alpha+n)\Gamma(\beta+1)},
\]
and Gautschi's inequality \citep{gautschi1959some} we may bound the Gamma ratio polynomially: for fixed
\(\beta>0\) there exists a constant \(C_\beta>0\) such that for all \(n\ge 1\)
\[
\frac{\Gamma(\alpha+n+\beta)}{\Gamma(\alpha+n)}\le C_\beta(\alpha+n)^{\beta}.
\]
Hence we obtain
\[
\begin{aligned}
&\Pi\!\left(s_{\vx^*}\le 1-\varepsilon\mid H_n\right)\\
&\qquad\le \frac{(\alpha+n+\beta)^{\beta}}{\Gamma(\beta+1)}\,
(1-\varepsilon)^{\alpha+n-1}.
\end{aligned}
\]
Finally, since $(1-\varepsilon)^{k}\le e^{-k\varepsilon}$ for all $k>0$,
the exponential bound in \eqref{eq:basin_tail} follows.
\end{proof}

\subsection{Spike-and-Slab Prior}
This subsection collects the moment bound used to prove the spike-and-slab
posterior rate in Proposition~\ref{prop:spike_slab_rate}.
Under the spike-and-slab prior, write
\[
\Pi = p\delta_1+(1-p)W,\qquad p>0,
\]
and let
\[
q_n:=\EE_W[s_{\vx^*}^n].
\]
Then Bayes' rule gives
\[
\Pi(s_{\vx^*}=1\mid H_n(\vx^*))
=
\frac{p}{p+(1-p)q_n}.
\]
The remaining task is to bound $q_n$ under a tail condition on $W$ near $1$.

\begin{lemma}[Moment bound from near--$1$ tail]\label{lem:master_moment}
Let $S$ be a random variable taking values in $[0,1]$ with law $W$, and define
\[
F(u) := W([1-u,1]) = \PP(S \ge 1-u)
\]
where $u \in [0,1]$. Then for any $\delta \in (0,1)$ and any integer $n \ge 1$,
\[
\EE_W[S^n]
\;\le\;
n \int_0^\delta e^{-(n-1)u} F(u)\,du
\;+\;
(1-\delta)^n.
\]
\end{lemma}
\phantomsection\label{proof:master_moment}
\begin{proof}[Proof of Lemma~\ref{lem:master_moment}]
Since $0 \le S \le 1$, we may write
\[
\EE[S^n]
=\int_0^1 \PP(S^n \ge t)\,dt.
\]
Because the map $x \mapsto x^n$ is increasing on $[0,1]$, we have
\[
\PP(S^n \ge t) = \PP(S \ge t^{1/n}).
\]
Making the change of variables $t = u^n$ (so $dt = n u^{n-1} du$) yields
\[
\EE[S^n]
=\int_0^1 n u^{n-1} \PP(S \ge u)\,du.
\]
Now set $u = 1 - r$ to obtain
\[
\begin{aligned}
\EE[S^n]
&=n \int_0^1 (1-r)^{n-1}
\PP(S \ge 1-r)\,dr\\
&=n \int_0^1 (1-r)^{n-1}F(r)\,dr.
\end{aligned}
\]
Using the elementary bound $(1-r)^{n-1} \le e^{-(n-1)r}$ for $r \in [0,1]$ and splitting the
integral at $\delta \in (0,1)$, we obtain
\[
\begin{aligned}
\EE[S^n]
&\le n \int_0^\delta e^{-(n-1)r}F(r)\,dr\\
&\qquad{}+n \int_\delta^1 (1-r)^{n-1}dr.
\end{aligned}
\]
The second term can be computed explicitly:
\[
n \int_\delta^1 (1-r)^{n-1} dr = (1-\delta)^n.
\]
Combining the bounds completes the proof.
\end{proof}

\begin{corollary}\label{cor:poly_tail}
Suppose there exist constants $C>0$, $\gamma>0$, and $\delta>0$ such that
\[
F(u) \le C u^\gamma \qquad \forall\, u \in (0,\delta].
\]
Then there exists $C'>0$ such that for all $n$ large enough,
\[
\EE_W[S^n] \le C' n^{-\gamma}.
\]
\end{corollary}
\phantomsection\label{proof:cor_poly_tail}
\begin{proof}[Proof of Corollary~\ref{cor:poly_tail}]
By Lemma~\ref{lem:master_moment},
\[
\begin{aligned}
\EE_W[S^n]
&\le n\int_0^\delta e^{-(n-1)u}F(u)\,du\\
&\qquad{}+(1-\delta)^n.
\end{aligned}
\]
Using the assumed tail bound $F(u)\le C u^\gamma$ on $(0,\delta]$, we obtain
\[
\begin{aligned}
\EE_W[S^n]
&\le Cn\int_0^\delta e^{-(n-1)u}u^\gamma\,du\\
&\qquad{}+(1-\delta)^n\\
&\le Cn\int_0^\infty e^{-(n-1)u}u^\gamma\,du\\
&\qquad{}+(1-\delta)^n.
\end{aligned}
\]
The remaining integral is the Gamma integral:
\[
\begin{aligned}
&\int_0^\infty e^{-(n-1)u}u^\gamma\,du\\
&\qquad=(n-1)^{-(\gamma+1)}
\int_0^\infty e^{-v}v^\gamma\,dv\\
&\qquad=(n-1)^{-(\gamma+1)}\Gamma(\gamma+1),
\end{aligned}
\]
where we used the change of variables $v=(n-1)u$. Hence
\[
\begin{aligned}
\EE_W[S^n]
&\le C\,\Gamma(\gamma+1)
\frac{n}{(n-1)^{\gamma+1}}\\
&\qquad{}+(1-\delta)^n.
\end{aligned}
\]
For all $n\ge 2$, $\frac{n}{(n-1)^{\gamma+1}}\le 2^{\gamma+1}n^{-\gamma}$, and also
$(1-\delta)^n\le e^{-\delta n}\le n^{-\gamma}$ for all $n$ large enough. Therefore there exists
a constant $C'>0$ such that $\EE_W[S^n]\le C' n^{-\gamma}$ for all sufficiently large $n$.
\end{proof}

\phantomsection\label{proof:spike_slab_rate}
\propSpikeSlabRate*
\begin{proof}[Proof of Proposition~\ref{prop:spike_slab_rate}]
Let $q_n:=\EE_W[s_{\vx^*}^n]$. By Corollary~\ref{cor:poly_tail}, the stated
tail condition on $W$ implies $q_n=O(n^{-\gamma})$. From the posterior formula
\eqref{eq:z=1_posterior},
\[
\begin{aligned}
1-\Pi(s_{\vx^*}=1\mid H_n(\vx^*))
&=1-\frac{p}{p+(1-p)q_n}\\
&=\frac{(1-p)q_n}{p+(1-p)q_n}\\
&\le \frac{1-p}{p}\,q_n,
\end{aligned}
\]
where $p>0$ by the spike-and-slab prior definition. Hence
$1-\Pi(s_{\vx^*}=1\mid H_n(\vx^*))=O(n^{-\gamma})$.
\end{proof}

\subsection{Mixture of Finite Models}
This subsection collects the MFM bounds used to prove
Theorem~\ref{thm:K1_tight_rate_fixed}.

\begin{lemma}[Two-sided bounds for $L_k(n)$ (fixed $\vx^*$)]\label{lem:Lk_twosided_fixed}
Fix $\alpha>0$ as part of the data-generating model $\PP$, and let
$\delta_k=\alpha(1-\tfrac1k)$. For each $k\ge 1$ and $n\ge 1$, define
\[
L_k(n)\;:=\;\PP(H_n(\vx^*)\mid K=k).
\]
This is well defined: under the symmetric prior, $L_k(n)$ does not depend on
the choice of $\vx^*$. Then for all $n\ge 1$,
\begin{align}
\frac1k\left(\frac{\alpha/k}{\,n-1+\alpha/k\,}\right)^{\delta_k}
\;\le\;
L_k(n)
\;\le\;
\frac1k\left(\frac{1+\alpha}{\,n+\alpha\,}\right)^{\delta_k}.
\label{eq:Lk_twosided_fixed}
\end{align}
\end{lemma}
\begin{proof}\label{proof:Lk_twosided_fixed}
Condition on $K=k$. Under the MFM prior, the basin-mass vector satisfies
\[
\begin{aligned}
\vbasin=(s_1,\dots,s_k)
&\sim\operatorname{Dirichlet}(a,\dots,a),\\
a&=\alpha/k,
\qquad \sum_{i=1}^k a=\alpha.
\end{aligned}
\]
Identifying $\vx^*$ with (without loss of generality) label $1$, conditional on $\vbasin$ we have
\[
\begin{aligned}
\PP(H_n(\vx^*)\mid \vbasin,K=k)
&=\PP(y_1=\cdots=y_n=1\mid \vbasin)\\
&=s_1^{n}.
\end{aligned}
\]
Taking expectation over $\vbasin$ gives
\[
L_k(n)=\PP(H_n(\vx^*)\mid K=k)=\EE[s_1^{n}].
\]
By the Dirichlet moment formula,
\[
\EE[s_1^{n}]
=\frac{\Gamma(\alpha)\Gamma(a+n)}{\Gamma(a)\Gamma(\alpha+n)}.
\]
Thus
\begin{equation}\label{eq:Lk_gamma_form_fixed}
L_k(n)
=\frac{\Gamma(\alpha)\Gamma(a+n)}{\Gamma(a)\Gamma(\alpha+n)}.
\end{equation}

\emph{Step 1: Gamma form to product form.}
Using $\Gamma(x+n)=\Gamma(x+1)\prod_{j=1}^{n-1}(x+j)$, we have
\[
\begin{aligned}
\Gamma(a+n)
&=\Gamma(a+1)\prod_{j=1}^{n-1}(a+j),\\
\Gamma(\alpha+n)
&=\Gamma(\alpha+1)\prod_{j=1}^{n-1}(\alpha+j).
\end{aligned}
\]
Substituting into \eqref{eq:Lk_gamma_form_fixed} yields
\[
L_k(n)
=\frac{\Gamma(\alpha)\Gamma(a+1)}{\Gamma(a)\Gamma(\alpha+1)}
\prod_{j=1}^{n-1}\frac{j+a}{j+\alpha}.
\]
Using $\Gamma(a+1)=a\Gamma(a)$ and $\Gamma(\alpha+1)=\alpha\Gamma(\alpha)$,
the prefactor simplifies to $\frac{a}{\alpha}=\frac1k$, hence
\begin{equation}\label{eq:Lk_product_form_fixed}
\begin{aligned}
L_k(n)
&=\frac1k\prod_{j=1}^{n-1}\frac{j+\alpha/k}{j+\alpha}\\
&=\frac1k\prod_{j=1}^{n-1}
\left(1-\frac{\delta_k}{j+\alpha}\right),\\
\delta_k&:=\alpha\Bigl(1-\frac1k\Bigr).
\end{aligned}
\end{equation}

\emph{Upper bound.}
Using $\log(1-t)\le -t$ for $t\in(0,1)$ and \eqref{eq:Lk_product_form_fixed},
\[
\begin{aligned}
\log L_k(n)
&=-\log k
+\sum_{j=1}^{n-1}
\log\!\left(1-\frac{\delta_k}{j+\alpha}\right)\\
&\le -\log k
-\delta_k\sum_{j=1}^{n-1}\frac{1}{j+\alpha}.
\end{aligned}
\]
Since $x\mapsto 1/(x+\alpha)$ is decreasing,
\[
\sum_{j=1}^{n-1}\frac{1}{j+\alpha}\ge \int_{1}^{n}\frac{dx}{x+\alpha}
=\log\!\left(\frac{n+\alpha}{1+\alpha}\right).
\]
Exponentiating gives
\[
L_k(n)\le \frac1k\left(\frac{1+\alpha}{n+\alpha}\right)^{\delta_k}.
\]

\emph{Lower bound.}
Using $\log(1-t)\ge -\frac{t}{1-t}$ for $t\in(0,1)$,
\[
\log\!\left(1-\frac{\delta_k}{j+\alpha}\right)
\ge -\frac{\delta_k}{j+\alpha-\delta_k}
=-\frac{\delta_k}{j+\alpha/k}.
\]
Summing and using \eqref{eq:Lk_product_form_fixed} yields
\[
\log L_k(n)\ge -\log k-\delta_k\sum_{j=1}^{n-1}\frac{1}{j+\alpha/k}.
\]
Since $x\mapsto 1/(x+\alpha/k)$ is decreasing,
\[
\begin{aligned}
\sum_{j=1}^{n-1}\frac{1}{j+\alpha/k}
&\le \int_{0}^{n-1}\frac{dx}{x+\alpha/k}\\
&=\log\!\left(\frac{n-1+\alpha/k}{\alpha/k}\right).
\end{aligned}
\]
Exponentiating gives
\[
L_k(n)\ge \frac1k\left(\frac{\alpha/k}{n-1+\alpha/k}\right)^{\delta_k}.
\]
\end{proof}

\begin{lemma}[Monotonicity of $L_k(n)$]\label{lem:Lk_monotone_fixed}
For fixed $\alpha>0$ and $n\ge 1$, the map $k\mapsto L_k(n)$ is decreasing on
$\{1,2,\dots\}$.
\end{lemma}
\begin{proof}
Write $a_k=\alpha/k$. By the Dirichlet moment formula,
\[
L_k(n)=\frac{\Gamma(\alpha)\Gamma(a_k+n)}{\Gamma(a_k)\Gamma(\alpha+n)}.
\]
The map $a\mapsto \Gamma(a+n)/\Gamma(a)$ is increasing on $(0,\infty)$, since
\[
\begin{aligned}
\frac{d}{da}\log\!\left(\frac{\Gamma(a+n)}{\Gamma(a)}\right)
&=\psi(a+n)-\psi(a)\\
&=\sum_{j=0}^{n-1}\frac{1}{a+j}>0.
\end{aligned}
\]
Since $k\mapsto a_k$ is decreasing, $k\mapsto L_k(n)$ is decreasing.
\end{proof}

\thmKOneTightRateFixed*
\begin{proof}\label{proof:K1_tight_fixed}
Write $N_n:=\sum_{k\ge 2}\pi_k L_k(n)$, where $L_k(n)=\PP(H_n(\vx^*)\mid K=k)$ is as in
Lemma~\ref{lem:Lk_twosided_fixed}. Then
\[
\begin{aligned}
1-\PP(K=1\mid H_n(\vx^*))
&=\frac{N_n}{\pi_1+N_n}.
\end{aligned}
\]

\emph{Upper bound.}
Since $\pi_1+N_n\ge \pi_1$,
\[
\begin{aligned}
1-\PP(K=1\mid H_n(\vx^*))
&\le \frac{N_n}{\pi_1}\\
&=\sum_{k\ge 2}\frac{\pi_k}{\pi_1}L_k(n).
\end{aligned}
\]
By Lemma~\ref{lem:Lk_monotone_fixed}, $L_k(n)\le L_2(n)$ for all $k\ge 2$, so
\[
1-\PP(K=1\mid H_n(\vx^*))\le \frac{1-\pi_1}{\pi_1}\,L_2(n).
\]
Apply the upper bound on $L_2(n)$ from Lemma~\ref{lem:Lk_twosided_fixed} with $k=2$:
since $\delta_2=\alpha/2$,
\[
L_2(n)\le \frac12\left(\frac{1+\alpha}{n+\alpha}\right)^{\alpha/2}.
\]

\emph{Lower bound.}
Since $\pi_1+N_n\le \pi_1+\sum_{k\ge 2}\pi_k=1$, we have
\[
\begin{aligned}
1-\PP(K=1\mid H_n(\vx^*))
&=\frac{N_n}{\pi_1+N_n}\\
&\ge N_n\\
&\ge \pi_2 L_2(n).
\end{aligned}
\]
Apply the lower bound on $L_2(n)$ from Lemma~\ref{lem:Lk_twosided_fixed} with $k=2$:
\[
L_2(n)\ge \frac12\left(\frac{\alpha/2}{n-1+\alpha/2}\right)^{\alpha/2}.
\]

Finally, both bounds are of order $n^{-\alpha/2}$ as $n\to\infty$, hence
$1-\PP(K=1\mid H_n(\vx^*))=\Theta(n^{-\alpha/2})$.
\end{proof}

\section{Application}\label{appendix:application}

This appendix collects two constructions used in Section~\ref{sec:applications}: a smooth compact state space containing the relevant price box and a Banach structure on the parametrized class of dynamics.

\subsection{A smooth compact price domain}\label{appendix:application-smooth-domain}

Let $0<m<M<\infty$ be lower and upper bounds such that the equilibrium prices lie in $[m,M]^d\subset(0,\infty)^d$. Such bounds are supplied by \citet{skerlos2010fixed} for the mixed-logit pricing problem considered in Section~\ref{sec:applications}. We construct a compact smooth superset of this box that remains inside the positive orthant.

Fix an integer $p\ge 1$, let $c=(m+M)/2$, and define
\[
r=\frac{M-m}{2}\,d^{1/(2p)}.
\]
Set
\[
X_p=\left\{\vx\in\R^d:\;\sum_{i=1}^d \left(\frac{\vx_i-c}{r}\right)^{2p}\le 1\right\}.
\]
Then $X_p$ is compact and convex, and $[m,M]^d\subseteq X_p$. Indeed, if $\vx\in[m,M]^d$, then $|\vx_i-c|\le (M-m)/2$ for every $i$, so
\[
\begin{aligned}
\sum_{i=1}^d \left(\frac{\vx_i-c}{r}\right)^{2p}
&\le d\left(\frac{(M-m)/2}{r}\right)^{2p}\\
&=1.
\end{aligned}
\]

Moreover, $X_p$ has $C^\infty$ boundary. Its boundary is the level set of
\[
h(\vx)=\sum_{i=1}^d \left(\frac{\vx_i-c}{r}\right)^{2p},
\]
and $\nabla h(\vx)\ne 0$ on $h^{-1}(1)$ because $h(\vx)=1$ implies at least one coordinate satisfies $\vx_i\ne c$. Finally, for large enough $p$, $c-r>0$, since $d^{1/(2p)}\to 1$ as $p\to\infty$ and $c>(M-m)/2$. Hence $X_p\subset(0,\infty)^d$ for such $p$.

\subsection{A transported Banach structure}\label{appendix:application-banach}

Let
\[
\mathcal{Q}
=\left\{\vQ_{(\vtheta,\vomega)}:(\vtheta,\vomega)\in
\R^{d_C}\times\R^{d_\zeta}\right\}
\]
be the class of dynamics generated by the chosen parametrization of marginal costs and markups,
\[
\vQ_{(\vtheta,\vomega)}(\vx)
=\vc_{\vtheta}+\vzeta_{\vomega}(\vx)-\vx.
\]
Assume the parametrization $(\vtheta,\vomega)\mapsto \vQ_{(\vtheta,\vomega)}$ is injective. Transport the vector-space operations from $\R^{d_C}\times\R^{d_\zeta}$ to $\mathcal{Q}$ by defining
\[
\begin{aligned}
\vQ_{(\vtheta_1,\vomega_1)}\oplus \vQ_{(\vtheta_2,\vomega_2)}
&:=\vQ_{(\vtheta_1+\vtheta_2,\vomega_1+\vomega_2)},\\
\lambda\odot \vQ_{(\vtheta,\vomega)}
&:=\vQ_{(\lambda\vtheta,\lambda\vomega)}.
\end{aligned}
\]
Equip $\mathcal{Q}$ with the transported norm
\[
\|\vQ_{(\vtheta,\vomega)}\|_{\mathcal Q}
=\|(\vtheta,\vomega)\|_2.
\]
The injectivity assumption makes these operations and this norm well defined. With this structure, the parametrization is a linear isometry from $\R^{d_C}\times\R^{d_\zeta}$ onto $\mathcal{Q}$. Since $\R^{d_C}\times\R^{d_\zeta}$ is finite-dimensional and complete, $\mathcal{Q}$ is a Banach space.

\end{document}